\documentclass[12pt]{article}    

\usepackage[margin=0.75in]{geometry} 
\usepackage{amsmath,amssymb,amsthm}

\newtheorem{theorem}{Theorem}[section]

\newtheorem{lemma}[theorem]{Lemma}

\theoremstyle{remark}

\newcommand{\al}{\alpha}

\newcommand{\be}{\begin{equation}}
\newcommand{\ee}{\end{equation}}

\newcommand{\bea}{\begin{eqnarray}}
\newcommand{\eea}{\end{eqnarray}}

\numberwithin{equation}{section}
\begin{document}

\title{Linear Statistics of Random Matrix Ensembles at the Spectrum Edge Associated with the Airy Kernel}
\author{Chao Min\thanks{School of Mathematical Sciences, Huaqiao University, Quanzhou 362021, China; e-mail: chaomin@hqu.edu.cn}\: and Yang Chen\thanks{Correspondence to: Yang Chen, Department of Mathematics, Faculty of Science and Technology, University of Macau, Macau, China; e-mail: yangbrookchen@yahoo.co.uk}}


\date{\today}
\maketitle
\begin{abstract}
In this paper, we study the large $N$ behavior of the moment-generating function (MGF) of the linear statistics of $N\times N$ Hermitian matrices in the Gaussian unitary, symplectic, orthogonal ensembles (GUE, GSE, GOE) and Laguerre unitary, symplectic, orthogonal ensembles (LUE, LSE, LOE) at the edge of the spectrum. From the finite $N$ Fredholm determinant expression of the MGF of the linear statistics \cite{Min201601}, we find the large $N$ asymptotics of the MGF associated with the Airy kernel in these Gaussian and Laguerre ensembles. Then we obtain the mean and variance of the suitably scaled linear statistics. We show that there is an equivalence between the large $N$ behavior of the MGF of the scaled linear statistics in Gaussian and Laguerre ensembles, which leads to the statistical equivalence between the mean and variance of suitably scaled linear statistics in Gaussian and Laguerre ensembles. In the end, we use the Coulomb fluid method to obtain the mean and variance of another type of linear statistics in GUE, which reproduces the result of Basor and Widom \cite{Basor1999}.
\end{abstract}

$\mathbf{Keywords}$: Random matrix ensembles; Linear statistics; Moment-generating function;

Airy kernel; Mean and variance; Asymptotics.

$\mathbf{Mathematics\:\: Subject\:\: Classification\:\: 2010}$: 15B52, 47A53, 33C45

\section{Introduction and Preliminaries}
In random matrix theory, the joint probability density function for the eigenvalues $\{x_j\}_{j=1}^{N}$ of $N\times N$ Hermitian matrices from an
unitary ensemble ($\beta=2$), symplectic ensemble ($\beta=4$) or orthogonal ensemble ($\beta=1$) is given by \cite{Mehta}
\be\label{jpdf}
P^{(\beta)}(x_{1},x_{2},\ldots,x_{N}):=C_{N}^{(\beta)}\prod_{1\leq j<k\leq N}\left|x_{j}-x_{k}\right|^{\beta}\prod_{j=1}^{N}w(x_{j}).
\ee
Here $w(x)$ is a weight function or a probability density supported on $[a,b]$, such that all the moments of $w$, namely, $\int_{a}^{b}x^j w(x)dx,\;\;j=0,1,2,\ldots$ exist. $C_{N}^{(\beta)}$ is the normalization constant so that
$$
\int_{[a,b]^N}P^{(\beta)}(x_{1},x_{2},\ldots,x_{N})\prod_{j=1}^{N}dx_{j}=1.
$$

Linear statistics is defined by the sum of a one variable function evaluated at the eigenvalues of random variables \cite{Chen1994}. The moment-generating function (MGF) of the linear statistics $\sum_{j=1}^{N}F(x_{j})$ is given by the following mathematical expectation with respect to the joint probability density function (\ref{jpdf}),
\be\label{mgf}
\mathbb{E}\left({\rm e}^{-\lambda\:\sum_{j=1}^{N}F(x_j)}\right)
:=\frac{\int_{[a,b]^{N}}\prod_{1\leq j<k\leq N}\left|x_{j}-x_{k}\right|^{\beta}\prod_{j=1}^{N}w(x_{j})\:{\rm e}^{-\lambda\:F(x_j)}dx_{j}}{\int_{[a,b]^{N}}\prod_{1\leq j<k\leq N}\left|x_{j}-x_{k}\right|^{\beta}\prod_{j=1}^{N}w(x_{j})dx_{j}},
\ee
where $\lambda$ is a parameter. Linear statistics is of great interest in various applications in random matrix theory. For example, in the situation $\beta=2$, the Shannon capacity is characterized by the MGF of the linear statistics $\sum_{j=1}^{N}-\ln\left(1+\frac{x_{j}}{t}\right)$ in the single-user multiple-input multiple-output (MIMO) wireless communication systems \cite{Chen2012}; the MGF of the center of mass of the eigenvalues corresponds to the linear statistics $\sum_{j=1}^{N}x_{j}$ in unitary ensembles \cite{Zhan}; the MGF of the Wigner delay time is associated with the linear statistics $\sum_{j=1}^{N}\frac{1}{x_{j}}$ \cite{Texier}. The $\beta=1$ case plays an important role in problems in multivariate statistics; see Johnstone \cite{Johnstone}. Although the $\beta=4$ case has less statistical applications, it is of interest in problems related to quantum transport in disordered systems \cite{Beenakker1,Beenakker2}.

The framework for the $\beta=1$ and $\beta=4$ ensembles was laid down by Widom \cite{Widom} and Tracy and Widom \cite{Tracy1998}, based on de Bruijn's formulas \cite{deBruijn}. In such a formalism the expectation of the generating function is cast as Fredholm determinant involving matrix kernel. If the background weight is normal density, amenable expression can be found, and large $N$ results could be obtained. In the situation where one travels to the bulk of the spectrum, the kernel in the Fredholm determinant is the sine kernel. At the edge of the spectrum, the eigenvalue density (or one point function) is $\sim \frac{\sqrt{2b(b-x)}}{\pi},\; b=\sqrt{2N}$, the reproducing kernel is expressed in terms of the Airy function and its derivative. This Airy kernel is significant in the Tracy-Widom distribution, which describes the largest eigenvalue distribution in Gaussian unitary ensemble \cite{Tracy1994A,TW}.

In this paper, we write (\ref{mgf}) in the following form
\be
G_{N}^{(\beta)}(f):=C_{N}^{(\beta)}\int_{[a,b]^{N}}\prod_{1\leq j<k\leq N}\left|x_{j}-x_{k}\right|^{\beta}
\prod_{j=1}^{N}w(x_{j})\left[1+f(x_{j})\right]dx_{j},\label{gnb}
\ee
where $f(x)={\rm e}^{-\lambda F(x)}-1$. We assume $f(x)$ lies in the Schwartz space \cite{Stein} and $1+f(x)\neq 0$ over $[a,b]$.
For the unitary ensembles, Tracy and Widom \cite{Tracy1998} obtained the following result for $G_{N}^{(2)}(f)$, which can be expressed as a Fredholm determinant.
\begin{lemma}\label{lemma1}
Let
$$
\varphi_{j}(x):=P_{j}(x)\sqrt{w(x)},\;\;j=0,1,2,\ldots,
$$
where $P_{j}(x),\;\;j=0,1,2,\ldots$ are the polynomials of degree $j$ orthonormal with respect to the weight $w(x)$,
$$
\int_{a}^{b}P_{j}(x)P_{k}(x)w(x)dx=\delta_{jk},\;j, k=0, 1, 2, \ldots.
$$
Then
$$
G_{N}^{(2)}(f)=\det\left(I+K_{N}^{(2)}f\right),
$$
where $K_{N}^{(2)}$ is an operator on $L^2[a,b]$ with kernel
$$
K_{N}^{(2)}(x,y):=\sum_{j=0}^{N-1}\varphi_{j}(x)\varphi_{j}(y),
$$
and $f$ denotes the operator, multiplication by $f$. That is, for a function $g\in L^2[a,b]$,
$$
\left(K_{N}^{(2)}f\right) g(x):=\int_{a}^{b}K_{N}^{(2)}(x,y)f(y)g(y)dy.
$$
\end{lemma}

Min and Chen \cite{Min201601} expressed $G_{N}^{(4)}(f)$ and $G_{N}^{(1)}(f)$ as Fredholm determinants based on the work of Dieng and Tracy \cite{Dieng} and Tracy and Widom \cite{Tracy1998}. For the $\beta=1$ case, we take $N$ to be even for simplicity. We state the results as the following two lemmas \cite{Min201601}.
\begin{lemma}\label{le2}
Define
$$
\psi_{j}(x):=\pi_{j}(x)\sqrt{w(x)},\;\;j=0,1,2,\ldots,
$$
where $\pi_{j}(x)$ is any polynomial of degree $j$, and
$$
M^{(4)}:=\left(\int_{a}^{b}\left(\psi_{j}(x)\psi_{k}'(x)-\psi_{j}'(x)\psi_{k}(x)\right)dx\right)_{j,k=0}^{2N-1}
$$
with its inverse denoted by
$$
\left(M^{(4)}\right)^{-1}=:(\mu_{jk})_{j,k=0}^{2N-1}.
$$
Then
$$
\left[G_{N}^{(4)}(f)\right]^{2}=\det\left(I+2K_{N}^{(4)}f-K_{N}^{(4)}\varepsilon f'\right),
$$
where $K_{N}^{(4)}$ and $\varepsilon$ are integral operators with kernel
$$
K_{N}^{(4)}(x,y):=-\sum_{j,k=0}^{2N-1}\mu_{jk}\psi_{j}(x)\psi_{k}'(y)
$$
and
$$
\varepsilon(x,y):=\frac{1}{2}\mathrm{sgn}(x-y),
$$
respectively.
\end{lemma}

\begin{lemma}\label{le3}
We assume that $N$ is even. Let
$$
\tilde{\psi}_{j}(x):=\pi_{j}(x)w(x),\;\;j=0,1,2,\ldots,
$$
where $\pi_{j}(x)$ is any polynomial of degree $j$, and
$$
M^{(1)}:=\left(\int_{a}^{b}\tilde{\psi}_{j}(x)\varepsilon\tilde{\psi}_{k}(x)dx\right)_{j,k=0}^{N-1}
$$
with its inverse denoted by
$$
\left(M^{(1)}\right)^{-1}=:(\nu_{jk})_{j,k=0}^{N-1}.
$$
Then
$$
\left[G_{N}^{(1)}(f)\right]^{2}
=\det\left(I+K_{N}^{(1)}(f^{2}+2f)-K_{N}^{(1)}\varepsilon f'-K_{N}^{(1)}f\varepsilon f'\right),
$$
where $K_{N}^{(1)}$ is an integral operator with kernel
$$
K_{N}^{(1)}(x,y):=\sum_{j,k=0}^{N-1}\nu_{jk}\varepsilon\tilde{\psi}_{j}(x)\tilde{\psi}_{k}(y).
$$
\end{lemma}

In this paper, we take $w(x)=\mathrm{e}^{-x^{2}},\; x\in \mathbb{R}$ and $w(x)=x^{\alpha}\mathrm{e}^{-x},\;\alpha>-1,\; x\in \mathbb{R}^{+}$. These are known as the Gaussian unitary, symplectic, orthogonal ensemble and Laguerre unitary, symplectic, orthogonal ensemble, respectively. The unitary ensemble is the simplest one; the symplectic ensemble and orthogonal ensemble are much more complicated than the unitary case. For the symplectic ensemble and orthogonal ensemble and also their relations to the unitary ensemble, see \cite{Adler2001,Adler2011,Dieng,Forrester2010,Mehta,Tracy1998,Widom} for reference. As far as we know, the linear statistics formula for $\beta=1$ and $\beta=4$ are produced in this paper for the first time that gives the mean and variance where the background is described by the edge behavior of the ensemble (hence the Airy kernel), except the original contribution of Basor and Widom for $\beta=2$ \cite{Basor1999}.

We introduce here some notations, which will be used in the following sections of this paper.
Let $\mathrm{Ai}(x)$ denote the Airy function, the first solution of $y''(x)-x\: y(x)=0$ \cite{Lebedev} (page 136-137).
Define
\be\label{bx}
B(x):=\int_{-\infty}^{x}\mathrm{Ai}(y)dy-\int_{x}^{\infty}\mathrm{Ai}(y)dy.
\ee
We can also write $B(x)$ in another form:
\bea
B(x)&=&\int_{-\infty}^{0}\mathrm{Ai}(y)dy-\int_{0}^{\infty}\mathrm{Ai}(y)dy+2\int_{0}^{x}\mathrm{Ai}(y)dy\nonumber\\
&=&\frac{1}{3}+2\int_{0}^{x}\mathrm{Ai}(y)dy,\nonumber
\eea
since we have the fact that \cite{Abramowitz} (page 449)
$$
\int_{-\infty}^{0}\mathrm{Ai}(y)dy=\frac{2}{3},\qquad \int_{0}^{\infty}\mathrm{Ai}(y)dy=\frac{1}{3}.
$$
Let $K(x,y)$ be the Airy kernel
\be\label{airy1}
K(x,y):=\frac{\mathrm{Ai}(x)\mathrm{Ai}'(y)-\mathrm{Ai}(y)\mathrm{Ai}'(x)}{x-y}.
\ee
When $x=y$,
\be\label{airy2}
K(x,x)=(\mathrm{Ai}'(x))^2-x\:\mathrm{Ai}^2(x).
\ee
The equality (\ref{airy2}) is obtained by taking the limit $y\rightarrow x$ from (\ref{airy1}) and using the property $\mathrm{Ai}''(x)=x\: \mathrm{Ai}(x)$.\\
We then define
\be\label{lxy}
L(x,y):=\int_{x}^{\infty}K(y,z)dz-\int_{-\infty}^{x}K(y,z)dz.
\ee
Finally, we mention that $\chi_{J}(x)$ is the indicator function defined on the interval $J$, namely,
$$
\chi_{J}(x)=\left\{
\begin{aligned}
&1,&x\in J;\\
&0,&x\notin J.
\end{aligned}
\right.
$$

The paper \cite{Min201601} studied the large $N$ behavior of the MGF of the linear statistics in Gaussian ensembles associated with the sine kernel and Laguerre ensembles associated with the Bessel kernel. This paper continues to study the large $N$ behavior of the MGF in these Gaussian and Laguerre ensembles associated with the Airy kernel, from which we obtain the mean and variance of the scaled linear statistics. The unitary case is the simplest one among them. We established the relation between the mean and variance of the scaled linear statistics in symplectic, orthogonal and unitary ensembles. We also show that as $N\rightarrow\infty$, the MGF of a suitably scaled linear statistics in the Gaussian ensembles and Laguerre ensembles are the same, which leads to the same mean and variance of the linear statistics between the Gaussian ensembles and Laguerre ensembles. For the problems on the mean and variance of linear statistics in unitary ensembles, see \cite{Basor1993,Basor1997,Chen1998,Min201602} for reference.

The rest of this paper is organized as follows. In Sec. 2, we study the large $N$ behavior of the MGF of the scaled linear statistics in Gaussian unitary, symplectic and orthogonal ensembles, respectively. From this we obtain the mean and variance of the scaled linear statistics in the three Gaussian ensembles. In Sec. 3, we repeat the development of Sec. 2, but for the Laguerre ensembles. In Sec. 4, we use the Coulomb fluid method to give an intuitive derivation for the $\beta=2$ result of Basor and Widom \cite{Basor1999}, and obtain the mean and variance of another type of linear statistics in GUE. The conclusion is given in Sec. 5.

\section{The Gaussian Ensembles}
\subsection{Gaussian Unitary Ensemble}
In the Gaussian case, $w(x)=\mathrm{e}^{-x^{2}},\; x\in \mathbb{R}$. From Lemma \ref{lemma1} we have
\be\label{phi}
\varphi_{j}(x)=\frac{1}{\pi^{\frac{1}{4}}2^{\frac{j}{2}}\sqrt{j!}}H_{j}(x)\mathrm{e}^{-\frac{x^{2}}{2}},\;\;j=0,1,2,\ldots,
\ee
where $H_{j}(x)$ are the Hermite polynomials of degree $j$.

We consider the large $N$ asymptotics of $G_{N}^{(2)}(f)$ in this subsection. It is well known that
\bea\label{logdet}
\log\det\left(I+K_{N}^{(2)}f\right)
&=&\mathrm{Tr}\log\left(I+K_{N}^{(2)}f\right)\nonumber\\
&=&\mathrm{Tr}K_{N}^{(2)}f-\frac{1}{2}\mathrm{Tr}\left(K_{N}^{(2)}f\right)^{2}+\frac{1}{3}\mathrm{Tr}\left(K_{N}^{(2)}f\right)^{3}-\cdots.
\eea
We state a theorem before our discussion.
\begin{theorem}\label{gue}
As $N\rightarrow\infty$,
$$
2^{-\frac{1}{2}}N^{-\frac{1}{6}}K_{N}^{(2)}\left(\sqrt{2N}+2^{-\frac{1}{2}}N^{-\frac{1}{6}}x,\sqrt{2N}+2^{-\frac{1}{2}}N^{-\frac{1}{6}}y\right)
=K(x,y)+O(N^{-\frac{1}{3}}),
$$
where $K(x,y)$ is the Airy kernel defined by (\ref{airy1}).
\end{theorem}
\begin{proof}
From the asymptotic formula of Hermite polynomial \cite{Szego} (page 201),
$$
\mathrm{e}^{-\frac{x^2}{2}}H_{n}(x)=2^{\frac{n}{2}+\frac{1}{4}}\pi^{\frac{1}{4}}(n!)^{\frac{1}{2}}n^{-\frac{1}{12}}
\left(\mathrm{Ai}(-3^{-\frac{1}{3}}t)+O(n^{-\frac{2}{3}})\right),
$$
where
$$
x=(2n+1)^{\frac{1}{2}}-2^{-\frac{1}{2}}3^{-\frac{1}{3}}n^{-\frac{1}{6}}t,
$$
we have
\be\label{asy}
\mathrm{e}^{-\frac{x^2}{2}}H_{n}(x)=\pi^{\frac{1}{4}}2^{\frac{n}{2}+\frac{1}{4}}(n!)^{\frac{1}{2}}n^{-\frac{1}{12}}
\left[\mathrm{Ai}\left(2^{\frac{1}{2}}n^{\frac{1}{6}}\left(x-\sqrt{2n+1}\right)\right)+O\left(n^{-\frac{2}{3}}\right)\right].
\ee
Using the Christoffel-Darboux formula,
$$
K_{N}^{(2)}(x,y)=\frac{\mathrm{e}^{-\frac{x^2}{2}}H_{N}(x)\mathrm{e}^{-\frac{y^2}{2}}H_{N-1}(y)
-\mathrm{e}^{-\frac{y^2}{2}}H_{N}(y)\mathrm{e}^{-\frac{x^2}{2}}H_{N-1}(x)}{\pi^{\frac{1}{2}}2^N(N-1)!(x-y)}.
$$
Replacing the variables $x$ by $\sqrt{2N}+2^{-\frac{1}{2}}N^{-\frac{1}{6}}x$ and $y$ by $\sqrt{2N}+2^{-\frac{1}{2}}N^{-\frac{1}{6}}y$, and using (\ref{asy}), we obtain the desired result after some elaborate computations.
\end{proof}
\noindent $\mathbf{Remark.}$ The above result was obtained by \cite{Bowick,Forrester,Moore}, but they did not show the order term. See also \cite{Tracy1994A} on the study of this Airy kernel.

We now use Theorem \ref{gue} to compute (\ref{logdet}) term by term as $N\rightarrow\infty$. We replace $f(x)$ by $f\left(2^{\frac{1}{2}}N^{\frac{1}{6}}\left(x-\sqrt{2N}\right)\right)$ in the following computations. The first term reads,
\bea
\mathrm{Tr}K_{N}^{(2)}f
&=&\int_{-\infty}^{\infty}K_{N}^{(2)}(x,x)f\left(2^{\frac{1}{2}}N^{\frac{1}{6}}\left(x-\sqrt{2N}\right)\right)dx\nonumber\\
&=&\int_{-\infty}^{\infty}2^{-\frac{1}{2}}N^{-\frac{1}{6}}K_{N}^{(2)}
\left(\sqrt{2N}+2^{-\frac{1}{2}}N^{-\frac{1}{6}}x,\sqrt{2N}+2^{-\frac{1}{2}}N^{-\frac{1}{6}}x\right)f(x)dx\nonumber\\
&=&\int_{-\infty}^{\infty}K(x,x)f(x)dx+O(N^{-\frac{1}{3}}).\nonumber
\eea
The second term,
\bea
\mathrm{Tr}\left(K_{N}^{(2)}f\right)^{2}
&=&\int_{-\infty}^{\infty}\int_{-\infty}^{\infty}K_{N}^{(2)}(x,y)f\left(2^{\frac{1}{2}}N^{\frac{1}{6}}\left(y-\sqrt{2N}\right)\right)K_{N}^{(2)}(y,x)
f\left(2^{\frac{1}{2}}N^{\frac{1}{6}}\left(x-\sqrt{2N}\right)\right)dx dy\nonumber\\
&=&\int_{-\infty}^{\infty}\int_{-\infty}^{\infty}2^{-\frac{1}{2}}N^{-\frac{1}{6}}K_{N}^{(2)}
\left(\sqrt{2N}+2^{-\frac{1}{2}}N^{-\frac{1}{6}}x,\sqrt{2N}+2^{-\frac{1}{2}}N^{-\frac{1}{6}}y\right)f(y)\nonumber\\
&\cdot&2^{-\frac{1}{2}}N^{-\frac{1}{6}}K_{N}^{(2)}\left(\sqrt{2N}+2^{-\frac{1}{2}}N^{-\frac{1}{6}}y,\sqrt{2N}+2^{-\frac{1}{2}}N^{-\frac{1}{6}}x\right)f(x)dx dy\nonumber\\
&=&\int_{-\infty}^{\infty}\int_{-\infty}^{\infty}K^{2}(x,y)f(x)f(y)dx dy+O(N^{-\frac{1}{3}}).\nonumber
\eea
It follows from (\ref{logdet}) that
\bea
\log\det\left(I+K_{N}^{(2)}f\right)
&=&\int_{-\infty}^{\infty}K(x,x)f(x)dx-\frac{1}{2}\int_{-\infty}^{\infty}\int_{-\infty}^{\infty}K^{2}(x,y)f(x)f(y)dx dy\nonumber\\
&+&\cdots+O(N^{-\frac{1}{3}}). \label{log1}
\eea

We proceed to study the mean and variance of the scaled linear statistics $\sum_{j=1}^{N}F\left(2^{\frac{1}{2}}N^{\frac{1}{6}}\left(x_{j}-\sqrt{2N}\right)\right)$, so we need to obtain the coefficients of $\lambda$ and $\lambda^{2}$ from (\ref{log1}). From the relation of $f(x)$ and $F(x)$ we know
\be\label{fx}
f(x)=-\lambda F(x)+\frac{\lambda^{2}}{2}F^{2}(x)-\cdots.
\ee
Substituting (\ref{fx}) into (\ref{log1}), we find
\bea
\log\det\left(I+K_{N}^{(2)}f\right)&=&-\lambda\int_{-\infty}^{\infty}K(x,x)F(x)dx\nonumber\\
&+&\frac{\lambda^{2}}{2}\bigg[\int_{-\infty}^{\infty}K(x,x)F^{2}(x)dx-\int_{-\infty}^{\infty}\int_{-\infty}^{\infty}K^{2}(x,y)F(x)F(y)dx dy\bigg]+\cdots+O(N^{-\frac{1}{3}}).\nonumber
\eea
Noting that $\log G_{N}^{(2)}(f)=\log\det\left(I+K_{N}^{(2)}f\right)$, we have the following theorem.
\begin{theorem}
Let $\mu_{N}^{(\mathrm{GUE})}$ and $\mathcal{V}_{N}^{(\mathrm{GUE})}$ be the mean and variance of the linear statistics\\
$\sum_{j=1}^{N}F\left(2^{\frac{1}{2}}N^{\frac{1}{6}}\left(x_{j}-\sqrt{2N}\right)\right)$. Then as $N\rightarrow\infty$,
\be
\mu_{N}^{(\mathrm{GUE})}=\int_{-\infty}^{\infty}K(x,x)F(x)dx+O(N^{-\frac{1}{3}}),\label{guem}
\ee
\be
\mathcal{V}_{N}^{(\mathrm{GUE})}=\int_{-\infty}^{\infty}K(x,x)F^{2}(x)dx-\int_{-\infty}^{\infty}\int_{-\infty}^{\infty}K^{2}(x,y)F(x)F(y)dx dy+O(N^{-\frac{1}{3}}),\label{guev}
\ee
where $K(x,y)$ is the Airy kernel defined by (\ref{airy1}).
\end{theorem}

\subsection{Gaussian Symplectic Ensemble}
In this case, $w(x)=\mathrm{e}^{-x^{2}},\;\;x\in \mathbb{R}$. Let
$$
\psi_{2j+1}(x):=\frac{1}{\sqrt{2}}\varphi_{2j+1}(x),\;\; \psi_{2j}(x):=-\frac{1}{\sqrt{2}}\varepsilon\varphi_{2j+1}(x),\;\;j=0,1,2,\ldots,
$$
where $\varphi_{j}(x)$ is given by (\ref{phi}). It follows that $M^{(4)}$ is the direct sum of the $N$ copies of
$
\begin{pmatrix}
0&1\\
-1&0
\end{pmatrix}
$
and $(M^{(4)})^{-1}=-M^{(4)}$ (see \cite{Dieng,Tracy1998}).
From Lemma \ref{le2}, we obtain the following result \cite{Min201601}.
\begin{theorem}
For the Gaussian symplectic ensemble, we have
$$
\left[G_{N}^{(4)}(f)\right]^{2}=\det(I+T_{\mathrm{GSE}}),
$$
where
$$
T_{\mathrm{GSE}}:=K_{2N+1}^{(2)}f-\frac{1}{2}K_{2N+1}^{(2)}\varepsilon f'+\sqrt{N+\frac{1}{2}}(\varepsilon\varphi_{2N+1})\otimes\varphi_{2N}f
+\frac{1}{2}\sqrt{N+\frac{1}{2}}(\varepsilon\varphi_{2N+1})\otimes(\varepsilon\varphi_{2N}) f',
$$
and $K_{2N+1}^{(2)}$ is an operator on $L^2(\mathbb{R})$ with kernel
$$
K_{2N+1}^{(2)}(x,y)=\sum_{j=0}^{2N}\varphi_{j}(x)\varphi_{j}(y).
$$
\end{theorem}
We also have the following expansion formula,
\be\label{logdet4}
\log\det(I+T_{\mathrm{GSE}})=\mathrm{Tr}\log(I+T_{\mathrm{GSE}})=\mathrm{Tr}\:T_{\mathrm{GSE}}-\frac{1}{2}\mathrm{Tr}\:T_{\mathrm{GSE}}^{2}
+\frac{1}{3}\mathrm{Tr}\:T_{\mathrm{GSE}}^{3}-\cdots.
\ee
Similarly as Theorem \ref{gue}, we have the following theorem.
\begin{theorem}\label{gse}
As $N\rightarrow\infty$,
$$
2^{-\frac{2}{3}}N^{-\frac{1}{6}}K_{2N+1}^{(2)}\left(\sqrt{4N}+2^{-\frac{2}{3}}N^{-\frac{1}{6}}x,\sqrt{4N}+2^{-\frac{2}{3}}N^{-\frac{1}{6}}y\right)
=K(x,y)+O(N^{-\frac{1}{3}}).
$$
\end{theorem}

\begin{theorem}\label{gse1}
As $N\rightarrow\infty$,
\be\label{eq1}
\varphi_{2N}\left(\sqrt{4N}+2^{-\frac{2}{3}}N^{-\frac{1}{6}}x\right)
=2^{\frac{1}{6}}N^{-\frac{1}{12}}\mathrm{Ai}(x)+O(N^{-\frac{3}{4}}),
\ee
\be\label{eq2}
\varphi_{2N+1}\left(\sqrt{4N}+2^{-\frac{2}{3}}N^{-\frac{1}{6}}x\right)
=2^{\frac{1}{6}}N^{-\frac{1}{12}}\mathrm{Ai}(x)+O(N^{-\frac{3}{4}}),
\ee
$$
\varepsilon\varphi_{2N}\left(\sqrt{4N}+2^{-\frac{2}{3}}N^{-\frac{1}{6}}x\right)
=2^{-\frac{3}{2}}N^{-\frac{1}{4}}B(x)+O(N^{-\frac{11}{12}}),
$$
$$
\varepsilon\varphi_{2N+1}\left(\sqrt{4N}+2^{-\frac{2}{3}}N^{-\frac{1}{6}}x\right)
=2^{-\frac{3}{2}}N^{-\frac{1}{4}}B(x)+O(N^{-\frac{11}{12}}),
$$
where $B(x)$ is defined by (\ref{bx}).
\end{theorem}
\begin{proof}
From the definition (\ref{phi}) and the asymptotic formula (\ref{asy}), we readily obtain (\ref{eq1}) and (\ref{eq2}).
It follows from the definition of $\varepsilon$ that
\bea
&&\varepsilon\varphi_{2N}\left(\sqrt{4N}+2^{-\frac{2}{3}}N^{-\frac{1}{6}}x\right)\nonumber\\
&=&\frac{1}{2}\left(\int_{-\infty}^{\sqrt{4N}+2^{-\frac{2}{3}}N^{-\frac{1}{6}}x}\varphi_{2N}(y)dy
-\int_{\sqrt{4N}+2^{-\frac{2}{3}}N^{-\frac{1}{6}}x}^{\infty}\varphi_{2N}(y)dy\right)\nonumber\\
&=&2^{-\frac{5}{3}}N^{-\frac{1}{6}}\left(\int_{-\infty}^{x}\varphi_{2N}\left(\sqrt{4N}+2^{-\frac{2}{3}}N^{-\frac{1}{6}}y\right)dy
-\int_{x}^{\infty}\varphi_{2N}\left(\sqrt{4N}+2^{-\frac{2}{3}}N^{-\frac{1}{6}}y\right)dy\right)\nonumber\\
&=&2^{-\frac{3}{2}}N^{-\frac{1}{4}}\left(\int_{-\infty}^{x}\mathrm{Ai}(y)dy
-\int_{x}^{\infty}\mathrm{Ai}(y)dy\right)+O(N^{-\frac{11}{12}})\nonumber\\
&=&2^{-\frac{3}{2}}N^{-\frac{1}{4}}B(x)+O(N^{-\frac{11}{12}}),\nonumber
\eea
where use has been made of (\ref{eq1}).
\\
Similarly, we find
\bea
&&\varepsilon\varphi_{2N+1}\left(\sqrt{4N}+2^{-\frac{2}{3}}N^{-\frac{1}{6}}x\right)\nonumber\\
&=&\frac{1}{2}\left(\int_{-\infty}^{\sqrt{4N}+2^{-\frac{2}{3}}N^{-\frac{1}{6}}x}\varphi_{2N+1}(y)dy
-\int_{\sqrt{4N}+2^{-\frac{2}{3}}N^{-\frac{1}{6}}x}^{\infty}\varphi_{2N+1}(y)dy\right)\nonumber\\
&=&2^{-\frac{5}{3}}N^{-\frac{1}{6}}\left(\int_{-\infty}^{x}\varphi_{2N+1}\left(\sqrt{4N}+2^{-\frac{2}{3}}N^{-\frac{1}{6}}y\right)dy
-\int_{x}^{\infty}\varphi_{2N+1}\left(\sqrt{4N}+2^{-\frac{2}{3}}N^{-\frac{1}{6}}y\right)dy\right)\nonumber\\
&=&2^{-\frac{3}{2}}N^{-\frac{1}{4}}\left(\int_{-\infty}^{x}\mathrm{Ai}(y)dy
-\int_{x}^{\infty}\mathrm{Ai}(y)dy\right)+O(N^{-\frac{11}{12}})\nonumber\\
&=&2^{-\frac{3}{2}}N^{-\frac{1}{4}}B(x)+O(N^{-\frac{11}{12}}).\nonumber
\eea
\end{proof}
Now we use Theorem \ref{gse} and Theorem \ref{gse1} to compute (\ref{logdet4}) as $N\rightarrow\infty$. We will change $f(x)$ to $f\left(2^{\frac{2}{3}}N^{\frac{1}{6}}\left(x-\sqrt{4N}\right)\right)$ in the following calculations. In this case, $f'(x)$ becomes\\ $2^{\frac{2}{3}}N^{\frac{1}{6}}f'\left(2^{\frac{2}{3}}N^{\frac{1}{6}}\left(x-\sqrt{4N}\right)\right)$. We consider $\mathrm{Tr}\:T_{\mathrm{GSE}}$ firstly,
$$
\mathrm{Tr}\:T_{\mathrm{GSE}}=\mathrm{Tr}\:K_{2N+1}f-\mathrm{Tr}\: \frac{1}{2}K_{2N+1}\varepsilon f'+\mathrm{Tr}\:\sqrt{N+\frac{1}{2}}
\left(\varepsilon\varphi_{2N+1}\otimes\varphi_{2N}f\right)
+\mathrm{Tr}\:\frac{1}{2}\sqrt{N+\frac{1}{2}}(\varepsilon\varphi_{2N+1}\otimes\varepsilon\varphi_{2N})f'.
$$
The first term reads,
\bea
\mathrm{Tr}\:K_{2N+1}^{(2)}f&=&\int_{-\infty}^{\infty}K_{2N+1}^{(2)}(x,x)f\left(2^{\frac{2}{3}}N^{\frac{1}{6}}\left(x-\sqrt{4N}\right)\right)dx\nonumber\\
&=&\int_{-\infty}^{\infty}2^{-\frac{2}{3}}N^{-\frac{1}{6}}K_{2N+1}^{(2)}\left(\sqrt{4N}+2^{-\frac{2}{3}}N^{-\frac{1}{6}}x,\sqrt{4N}
+2^{-\frac{2}{3}}N^{-\frac{1}{6}}x\right)f(x)dx\nonumber\\
&=&\int_{-\infty}^{\infty}K(x,x)f(x)dx+O(N^{-\frac{1}{3}}).\nonumber
\eea
The second term,
\bea
\mathrm{Tr}\: \frac{1}{2}K_{2N+1}^{(2)}\varepsilon f'&=&2^{-\frac{1}{3}}N^{\frac{1}{6}}\int_{-\infty}^{\infty}\int_{-\infty}^{\infty}K_{2N+1}^{(2)}(x,y)\varepsilon(y,x)
f'\left(2^{\frac{2}{3}}N^{\frac{1}{6}}\left(x-\sqrt{4N}\right)\right)dxdy\nonumber\\
&=&2^{-\frac{4}{3}}N^{\frac{1}{6}}\int_{-\infty}^{\infty}\left(\int_{x}^{\infty}K_{2N+1}^{(2)}(x,y)dy-\int_{-\infty}^{x}K_{2N+1}^{(2)}(x,y)dy\right)
f'\left(2^{\frac{2}{3}}N^{\frac{1}{6}}\left(x-\sqrt{4N}\right)\right)dx.\nonumber
\eea
Let
$$
x=\sqrt{4N}+2^{-\frac{2}{3}}N^{-\frac{1}{6}}u,\;\;y=\sqrt{4N}+2^{-\frac{2}{3}}N^{-\frac{1}{6}}v.
$$
Then
\bea
\mathrm{Tr}\: \frac{1}{2}K_{2N+1}^{(2)}\varepsilon f'&=&\frac{1}{4}\int_{-\infty}^{\infty}\bigg(\int_{u}^{\infty}2^{-\frac{2}{3}}N^{-\frac{1}{6}}K_{2N+1}^{(2)}
\left(\sqrt{4N}+2^{-\frac{2}{3}}N^{-\frac{1}{6}}u,\sqrt{4N}
+2^{-\frac{2}{3}}N^{-\frac{1}{6}}v\right)dv\nonumber\\
&-&\int_{-\infty}^{u}2^{-\frac{2}{3}}N^{-\frac{1}{6}}K_{2N+1}^{(2)}\left(\sqrt{4N}
+2^{-\frac{2}{3}}N^{-\frac{1}{6}}u,\sqrt{4N}+2^{-\frac{2}{3}}N^{-\frac{1}{6}}v\right)dv\bigg)f'(u)du\nonumber\\
&=&\frac{1}{4}\int_{-\infty}^{\infty}\left(\int_{u}^{\infty}K(u,v)dv-\int_{-\infty}^{u}K(u,v)dv\right)f'(u)du+O(N^{-\frac{1}{3}})\nonumber\\
&=&\frac{1}{4}\int_{-\infty}^{\infty}L(x,x)f'(x)dx+O(N^{-\frac{1}{3}}),\nonumber
\eea
where $L(x,y)$ is given by (\ref{lxy}).
\\
The third term,
\bea
\mathrm{Tr}\:\sqrt{N+\frac{1}{2}}
\left(\varepsilon\varphi_{2N+1}\otimes\varphi_{2N}f\right)&=&\int_{-\infty}^{\infty}\sqrt{N+\frac{1}{2}}\varepsilon\varphi_{2N+1}(x)\varphi_{2N}(x)
f\left(2^{\frac{2}{3}}N^{\frac{1}{6}}\left(x-\sqrt{4N}\right)\right)dx\nonumber\\
&=&2^{-\frac{2}{3}}N^{-\frac{1}{6}}\sqrt{N+\frac{1}{2}}\int_{-\infty}^{\infty}\varepsilon
\varphi_{2N+1}\left(\sqrt{4N}+2^{-\frac{2}{3}}N^{-\frac{1}{6}}x\right)\nonumber\\
&\cdot&\varphi_{2N}\left(\sqrt{4N}+2^{-\frac{2}{3}}N^{-\frac{1}{6}}x\right)f(x)dx\nonumber\\
&=&\frac{1}{4}\int_{-\infty}^{\infty}\mathrm{Ai}(x)B(x)f(x)dx+O(N^{-\frac{2}{3}}).\nonumber
\eea
The fourth term,
\bea
\mathrm{Tr}\:\frac{1}{2}\sqrt{N+\frac{1}{2}}(\varepsilon\varphi_{2N+1}\otimes\varepsilon\varphi_{2N})f'
&=&2^{-\frac{1}{3}}N^{\frac{1}{6}}\sqrt{N+\frac{1}{2}}\int_{-\infty}^{\infty}\varepsilon\varphi_{2N+1}(x)\varepsilon\varphi_{2N}(x)
f'\left(2^{\frac{2}{3}}N^{\frac{1}{6}}\left(x-\sqrt{4N}\right)\right)dx\nonumber\\
&=&\frac{1}{2}\sqrt{N+\frac{1}{2}}\int_{-\infty}^{\infty}\varepsilon
\varphi_{2N+1}\left(\sqrt{4N}+2^{-\frac{2}{3}}N^{-\frac{1}{6}}x\right)\nonumber\\
&\cdot&\varepsilon\varphi_{2N}\left(\sqrt{4N}+2^{-\frac{2}{3}}N^{-\frac{1}{6}}x\right)f'(x)dx\nonumber\\
&=&\frac{1}{16}\int_{-\infty}^{\infty}B^2(x)f'(x)dx+O(N^{-\frac{2}{3}}).\nonumber
\eea
So we obtain
\bea
\mathrm{Tr}\:T_{\mathrm{GSE}}&=&\int_{-\infty}^{\infty}K(x,x)f(x)dx-\frac{1}{4}\int_{-\infty}^{\infty}L(x,x)f'(x)dx
+\frac{1}{4}\int_{-\infty}^{\infty}\mathrm{Ai}(x)B(x)f(x)dx\nonumber\\
&+&\frac{1}{16}\int_{-\infty}^{\infty}B^2(x)f'(x)dx+O(N^{-\frac{1}{3}}).\label{trs}
\eea
We proceed to compute $\mathrm{Tr}\:T_{\mathrm{GSE}}^{2}$,
\bea
\mathrm{Tr}\:T_{\mathrm{GSE}}^{2}
&=&\mathrm{Tr}\:K_{2N+1}^{(2)}f K_{2N+1}^{(2)}f-\mathrm{Tr}\:K_{2N+1}^{(2)}f K_{2N+1}^{(2)}\varepsilon f'+\mathrm{Tr}\:\sqrt{4N+2}K_{2N+1}^{(2)}f(\varepsilon\varphi_{2N+1}
\otimes\varphi_{2N}f)\nonumber\\
&+&\mathrm{Tr}\:\sqrt{N+\frac{1}{2}}K_{2N+1}^{(2)}f(\varepsilon\varphi_{2N+1}\otimes\varepsilon\varphi_{2N})f'
+\mathrm{Tr}\:\frac{1}{4}K_{2N+1}^{(2)}\varepsilon f'K_{2N+1}^{(2)}\varepsilon f'\nonumber\\
&-&\mathrm{Tr}\:\sqrt{N+\frac{1}{2}}K_{2N+1}^{(2)}\varepsilon f'(\varepsilon\varphi_{2N+1}\otimes\varphi_{2N}f)
-\mathrm{Tr}\:\frac{1}{2}\sqrt{N+\frac{1}{2}}K_{2N+1}^{(2)}\varepsilon f'(\varepsilon\varphi_{2N+1}\otimes\varepsilon\varphi_{2N})f'\nonumber\\
&+&\mathrm{Tr}\:\left(N+\frac{1}{2}\right)(\varepsilon\varphi_{2N+1}\otimes\varphi_{2N}f)(\varepsilon\varphi_{2N+1}\otimes\varphi_{2N}f)
\nonumber\\
&+&\mathrm{Tr}\:\left(N+\frac{1}{2}\right)(\varepsilon\varphi_{2N+1}\otimes\varphi_{2N}f)(\varepsilon\varphi_{2N+1}\otimes\varepsilon
\varphi_{2N})f'\nonumber\\
&+&\mathrm{Tr}\:\frac{1}{4}\left(N+\frac{1}{2}\right)(\varepsilon\varphi_{2N+1}\otimes\varepsilon\varphi_{2N})f'(\varepsilon\varphi_{2N+1}
\otimes\varepsilon\varphi_{2N})f'.\nonumber
\eea
Similarly, we compute the ten traces one by one. We write down the result here without the detailed calculations:
\begin{small}
\bea\label{trs2}
&&\mathrm{Tr}\:T_{\mathrm{GSE}}^{2}\nonumber\\
&=&\int_{-\infty}^{\infty}\int_{-\infty}^{\infty}K^2(x,y)f(x)f(y)dxdy
-\frac{1}{2}\int_{-\infty}^{\infty}\int_{-\infty}^{\infty}K(x,y)L(x,y)f'(x)f(y)dxdy\nonumber\\
&+&\frac{1}{2}\int_{-\infty}^{\infty}\int_{-\infty}^{\infty}K(x,y)\mathrm{Ai}(x)B(y)f(x)f(y)dxdy
+\frac{1}{8}\int_{-\infty}^{\infty}\int_{-\infty}^{\infty}K(x,y)B(x)B(y)f(x)f'(y)dxdy\nonumber\\
&+&\frac{1}{16}\int_{-\infty}^{\infty}\int_{-\infty}^{\infty}L(x,y)L(y,x)f'(x)f'(y)dxdy
-\frac{1}{8}\int_{-\infty}^{\infty}\int_{-\infty}^{\infty}L(x,y)\mathrm{Ai}(y)B(x)f'(x)f(y)dxdy\nonumber\\
&-&\frac{1}{32}\int_{-\infty}^{\infty}\int_{-\infty}^{\infty}L(x,y)B(x)B(y)f'(x)f'(y)dxdy
+\frac{1}{16}\int_{-\infty}^{\infty}\int_{-\infty}^{\infty}\mathrm{Ai}(x)\mathrm{Ai}(y)B(x)B(y)f(x)f(y)dxdy\nonumber\\
&+&\frac{1}{32}\int_{-\infty}^{\infty}\int_{-\infty}^{\infty}\mathrm{Ai}(x)B(x)B^2(y)f(x)f'(y)dxdy
+\frac{1}{256}\int_{-\infty}^{\infty}\int_{-\infty}^{\infty}B^2(x)B^2(y)f'(x)f'(y)dxdy+O(N^{-\frac{1}{3}}).\nonumber\\
\eea
\end{small}

Proceeding as in the previous subsection, we replace $f(x)$ with $-\lambda F(x)+\frac{\lambda^2}{2}F^2(x)$, then $f'(x)$ becomes $-\lambda F'(x)+\lambda^2 F(x)F'(x)$. Substituting these into (\ref{trs}) and (\ref{trs2}), we finally find
\begin{small}
\bea
&&\log\det(I+T_{\mathrm{GSE}})\nonumber\\
&=&-\lambda\bigg\{\int_{-\infty}^{\infty}K(x,x)F(x)dx-\frac{1}{4}\int_{-\infty}^{\infty}L(x,x)F'(x)dx
+\frac{1}{4}\int_{-\infty}^{\infty}\mathrm{Ai}(x)B(x)F(x)dx
+\frac{1}{16}\int_{-\infty}^{\infty}B^2(x)F'(x)dx\bigg\}\nonumber\\
&+&\frac{\lambda^2}{2}\bigg\{\int_{-\infty}^{\infty}K(x,x)F^2(x)dx-\frac{1}{2}\int_{-\infty}^{\infty}L(x,x)F(x)F'(x)dx
+\frac{1}{4}\int_{-\infty}^{\infty}\mathrm{Ai}(x)B(x)F^2(x)dx\nonumber\\
&+&\frac{1}{8}\int_{-\infty}^{\infty}B^2(x)F(x)F'(x)dx
-\int_{-\infty}^{\infty}\int_{-\infty}^{\infty}K^2(x,y)F(x)F(y)dxdy\nonumber\\
&+&\frac{1}{2}\int_{-\infty}^{\infty}\int_{-\infty}^{\infty}K(x,y)L(x,y)F'(x)F(y)dxdy
-\frac{1}{2}\int_{-\infty}^{\infty}\int_{-\infty}^{\infty}K(x,y)\mathrm{Ai}(x)B(y)F(x)F(y)dxdy\nonumber\\
&-&\frac{1}{8}\int_{-\infty}^{\infty}\int_{-\infty}^{\infty}K(x,y)B(x)B(y)F(x)F'(y)dxdy
-\frac{1}{16}\int_{-\infty}^{\infty}\int_{-\infty}^{\infty}L(x,y)L(y,x)F'(x)F'(y)dxdy\nonumber\\
&+&\frac{1}{8}\int_{-\infty}^{\infty}\int_{-\infty}^{\infty}L(x,y)\mathrm{Ai}(y)B(x)F'(x)F(y)dxdy
+\frac{1}{32}\int_{-\infty}^{\infty}\int_{-\infty}^{\infty}L(x,y)B(x)B(y)F'(x)F'(y)dxdy\nonumber\\
&-&\frac{1}{16}\int_{-\infty}^{\infty}\int_{-\infty}^{\infty}\mathrm{Ai}(x)\mathrm{Ai}(y)B(x)B(y)F(x)F(y)dxdy
-\frac{1}{32}\int_{-\infty}^{\infty}\int_{-\infty}^{\infty}\mathrm{Ai}(x)B(x)B^2(y)F(x)F'(y)dxdy\nonumber\\
&-&\frac{1}{256}\int_{-\infty}^{\infty}\int_{-\infty}^{\infty}B^2(x)B^2(y)F'(x)F'(y)dxdy\bigg\}+O(N^{-\frac{1}{3}}).\nonumber
\eea
\end{small}
Noting that $\log\: G_{N}^{(4)}(f)=\frac{1}{2}\log\det(I+T_{\mathrm{GSE}})$, we obtain the following theorem.
\begin{theorem}
Denoting by $\mu_{N}^{(\mathrm{GSE})}$ and $\mathcal{V}_{N}^{(\mathrm{GSE})}$ the mean and variance of the linear statistics
$\sum_{j=1}^{N}F\left(2^{\frac{2}{3}}N^{\frac{1}{6}}\left(x_{j}-\sqrt{4N}\right)\right)$, we have as $N\rightarrow\infty$,
\begin{small}
$$
\mu_{N}^{(\mathrm{GSE})}=\frac{1}{2}\mu_{N}^{(\mathrm{GUE})}-\frac{1}{8}\int_{-\infty}^{\infty}L(x,x)F'(x)dx
+\frac{1}{8}\int_{-\infty}^{\infty}\mathrm{Ai}(x)B(x)F(x)dx+\frac{1}{32}\int_{-\infty}^{\infty}B^2(x)F'(x)dx+O(N^{-\frac{1}{3}}),
$$
\end{small}
\begin{small}
\bea
\mathcal{V}_{N}^{(\mathrm{GSE})}
&=&\frac{1}{2}\mathcal{V}_{N}^{(\mathrm{GUE})}-\frac{1}{4}\int_{-\infty}^{\infty}L(x,x)F(x)F'(x)dx
+\frac{1}{8}\int_{-\infty}^{\infty}\mathrm{Ai}(x)B(x)F^2(x)dx+\frac{1}{16}\int_{-\infty}^{\infty}B^2(x)F(x)F'(x)dx\nonumber\\
&+&\frac{1}{4}\int_{-\infty}^{\infty}\int_{-\infty}^{\infty}K(x,y)L(x,y)F'(x)F(y)dxdy
-\frac{1}{4}\int_{-\infty}^{\infty}\int_{-\infty}^{\infty}K(x,y)\mathrm{Ai}(x)B(y)F(x)F(y)dxdy\nonumber\\
&-&\frac{1}{16}\int_{-\infty}^{\infty}\int_{-\infty}^{\infty}K(x,y)B(x)B(y)F(x)F'(y)dxdy
-\frac{1}{32}\int_{-\infty}^{\infty}\int_{-\infty}^{\infty}L(x,y)L(y,x)F'(x)F'(y)dxdy\nonumber\\
&+&\frac{1}{16}\int_{-\infty}^{\infty}\int_{-\infty}^{\infty}L(x,y)\mathrm{Ai}(y)B(x)F'(x)F(y)dxdy
+\frac{1}{64}\int_{-\infty}^{\infty}\int_{-\infty}^{\infty}L(x,y)B(x)B(y)F'(x)F'(y)dxdy\nonumber\\
&-&\frac{1}{32}\int_{-\infty}^{\infty}\int_{-\infty}^{\infty}\mathrm{Ai}(x)\mathrm{Ai}(y)B(x)B(y)F(x)F(y)dxdy
-\frac{1}{64}\int_{-\infty}^{\infty}\int_{-\infty}^{\infty}\mathrm{Ai}(x)B(x)B^2(y)F(x)F'(y)dxdy\nonumber\\
&-&\frac{1}{512}\int_{-\infty}^{\infty}\int_{-\infty}^{\infty}B^2(x)B^2(y)F'(x)F'(y)dxdy+O(N^{-\frac{1}{3}}),\nonumber
\eea
\end{small}
where $\mu_{N}^{(\mathrm{GUE})}$ and $\mathcal{V}_{N}^{(\mathrm{GUE})}$ are given by (\ref{guem}) and (\ref{guev}), respectively.
\end{theorem}

\subsection{Gaussian Orthogonal Ensemble}
It is convenient in this case to choose $w(x)$ to be the \textit{square root} of the Gaussian weight,
$$
w(x)=\mathrm{e}^{-\frac{x^{2}}{2}},\;\;x\in \mathbb{R},
$$
and keep in mind that $N$ is even. Define
$$
\psi_{2n+1}(x):=\frac{d}{dx}\varphi_{2n}(x),\;\;\psi_{2n}(x):=\varphi_{2n}(x),\;\;n=0,1,2,\ldots,
$$
where $\varphi_{j}(x)$ is given by (\ref{phi}). It follows that $M^{(1)}$ is the direct sum of the $\frac{N}{2}$ copies of
$
\begin{pmatrix}
0&1\\
-1&0
\end{pmatrix}
$
and $(M^{(1)})^{-1}=-M^{(1)}$ (see \cite{Dieng,Tracy1998}).
From Lemma \ref{le3}, we obtain the following result \cite{Min201601}.
\begin{theorem}
For the Gaussian orthogonal ensemble, we have
$$
\left[G_{N}^{(1)}(f)\right]^{2}=\det(I+T_{\mathrm{GOE}}),
$$
where
\bea
T_{\mathrm{GOE}}:
&=&K_{N}^{(2)}(f^{2}+2f)-K_{N}^{(2)}\varepsilon f'-K_{N}^{(2)}f\varepsilon f'+\sqrt{\frac{N}{2}}(\varepsilon\varphi_{N})\otimes\varphi_{N-1}(f^{2}+2f)\nonumber\\
&+&\sqrt{\frac{N}{2}}(\varepsilon\varphi_{N})\otimes(\varepsilon\varphi_{N-1})f'-\sqrt{\frac{N}{2}}((\varepsilon\varphi_{N})
\otimes\varphi_{N-1})f\varepsilon f',\nonumber
\eea
and $K_{N}^{(2)}$ is an operator on $L^2(\mathbb{R})$ with kernel
$$
K_{N}^{(2)}(x,y)=\sum_{j=0}^{N-1}\varphi_{j}(x)\varphi_{j}(y).
$$
\end{theorem}
We also have
$$
\log\det(I+T_{\mathrm{GOE}})=\mathrm{Tr}\log(I+T_{\mathrm{GOE}})=\mathrm{Tr}\:T_{\mathrm{GOE}}-\frac{1}{2}\mathrm{Tr}\:T_{\mathrm{GOE}}^{2}+\cdots.
$$
Similarly as the previous subsection, we have the following results.
\begin{theorem}\label{goe}
As $N\rightarrow\infty$,
$$
\varphi_{N}\left(\sqrt{2N}+2^{-\frac{1}{2}}N^{-\frac{1}{6}}x\right)
=2^{\frac{1}{4}}N^{-\frac{1}{12}}\mathrm{Ai}(x)+O(N^{-\frac{3}{4}}),
$$
$$
\varphi_{N-1}\left(\sqrt{2N}+2^{-\frac{1}{2}}N^{-\frac{1}{6}}x\right)
=2^{\frac{1}{4}}N^{-\frac{1}{12}}\mathrm{Ai}(x)+O(N^{-\frac{3}{4}}),
$$
$$
\varepsilon\varphi_{N}\left(\sqrt{2N}+2^{-\frac{1}{2}}N^{-\frac{1}{6}}x\right)
=2^{-\frac{5}{4}}N^{-\frac{1}{4}}B(x)+O(N^{-\frac{11}{12}}),
$$
$$
\varepsilon\varphi_{N-1}\left(\sqrt{2N}+2^{-\frac{1}{2}}N^{-\frac{1}{6}}x\right)
=2^{-\frac{5}{4}}N^{-\frac{1}{4}}B(x)+O(N^{-\frac{11}{12}}),
$$
where $B(x)$ is given by (\ref{bx}).
\end{theorem}
In the computations below, we replace $f(x)$ by $f\left(2^{\frac{1}{2}}N^{\frac{1}{6}}\left(x-\sqrt{2N}\right)\right)$ and $f'(x)$ by\\ $2^{\frac{1}{2}}N^{\frac{1}{6}}f'\left(2^{\frac{1}{2}}N^{\frac{1}{6}}\left(x-\sqrt{2N}\right)\right)$.
Using Theorem \ref{gue} and Theorem \ref{goe} to compute $\mathrm{Tr}\:T_{\mathrm{GOE}}$ and $\mathrm{Tr}\:T_{\mathrm{GOE}}^2$ as $N\rightarrow\infty$, we obtain the following results:
\bea\label{trgoe}
\mathrm{Tr}\:T_{\mathrm{GOE}}&=&\int_{-\infty}^{\infty}K(x,x)(f^2(x)+2f(x))dx-\frac{1}{2}\int_{-\infty}^{\infty}L(x,x)f'(x)dx\nonumber\\
&-&\frac{1}{2}\int_{-\infty}^{\infty}dxf'(x)\int_{-\infty}^{\infty}(1-2\chi_{(-\infty,x)}(y))K(x,y)f(y)dy\nonumber\\
&+&\frac{1}{4}\int_{-\infty}^{\infty}\mathrm{Ai}(x)B(x)(f^2(x)+2f(x))dx+\frac{1}{8}\int_{-\infty}^{\infty}B^2(x)f'(x)dx\nonumber\\
&-&\frac{1}{8}\int_{-\infty}^{\infty}dx B(x)f'(x)\int_{-\infty}^{\infty}(1-2\chi_{(-\infty,x)}(y))\mathrm{Ai}(y)f(y)dy+O(N^{-\frac{1}{3}}),
\eea
\bea\label{trgoe2}
\mathrm{Tr}\:T_{\mathrm{GOE}}^2&=&\int_{-\infty}^{\infty}\int_{-\infty}^{\infty}K^2(x,y)(f^2(x)+2f(x))(f^2(y)+2f(y))dxdy\nonumber\\
&-&\int_{-\infty}^{\infty}\int_{-\infty}^{\infty}K(x,y)L(x,y)f'(x)(f^2(y)+2f(y))dxdy\nonumber\\
&+&\frac{1}{2}\int_{-\infty}^{\infty}\int_{-\infty}^{\infty}K(x,y)\mathrm{Ai}(x)B(y)(f^2(x)+2f(x))(f^2(y)+2f(y))dxdy\nonumber\\
&+&\frac{1}{4}\int_{-\infty}^{\infty}\int_{-\infty}^{\infty}K(x,y)B(x)B(y)f'(x)(f^2(y)+2f(y))dxdy\nonumber\\
&+&\frac{1}{4}\int_{-\infty}^{\infty}\int_{-\infty}^{\infty}L(x,y)L(y,x)f'(x)f'(y)dxdy\nonumber\\
&-&\frac{1}{4}\int_{-\infty}^{\infty}\int_{-\infty}^{\infty}L(x,y)\mathrm{Ai}(y)B(x)f'(x)(f^2(y)+2f(y))dxdy\nonumber\\
&-&\frac{1}{8}\int_{-\infty}^{\infty}\int_{-\infty}^{\infty}L(x,y)B(x)B(y)f'(x)f'(y)dxdy\nonumber\\
&+&\frac{1}{16}\int_{-\infty}^{\infty}\int_{-\infty}^{\infty}\mathrm{Ai}(x)\mathrm{Ai}(y)B(x)B(y)(f^2(x)+2f(x))(f^2(y)+2f(y))dxdy\nonumber\\
&+&\frac{1}{16}\int_{-\infty}^{\infty}\int_{-\infty}^{\infty}\mathrm{Ai}(x)B(x)B^2(y)(f^2(x)+2f(x))f'(y)dxdy\nonumber\\
&+&\frac{1}{64}\int_{-\infty}^{\infty}\int_{-\infty}^{\infty}B^2(x)B^2(y)f'(x)f'(y)dxdy+R+O(N^{-\frac{1}{3}}),
\eea
where $R$ contains the terms of integrals with integrands consisting of $f$, $f$, $f'$ or $f$, $f'$, $f'$. These lead to at least power 3 of $\lambda$ in the following discussions, and they will not affect the final results, so we need not write down the detailed results of $R$.

Similarly as the previous subsection, we replace $f(x)$ with $-\lambda F(x)+\frac{\lambda^2}{2}F^2(x)$ and $f'(x)$ with $-\lambda F'(x)+\lambda^2 F(x)F'(x)$. Substituting these into (\ref{trgoe}) and (\ref{trgoe2}), and we finally find
\begin{small}
\bea
&&\log\det(I+T_{\mathrm{GOE}})\nonumber\\
&=&-\lambda\bigg\{2\int_{-\infty}^{\infty}K(x,x)F(x)dx-\frac{1}{2}\int_{-\infty}^{\infty}L(x,x)F'(x)dx
+\frac{1}{2}\int_{-\infty}^{\infty}\mathrm{Ai}(x)B(x)F(x)dx
+\frac{1}{8}\int_{-\infty}^{\infty}B^2(x)F'(x)dx\bigg\}\nonumber\\
&+&\frac{\lambda^2}{2}\bigg\{4\int_{-\infty}^{\infty}K(x,x)F^2(x)dx-\int_{-\infty}^{\infty}L(x,x)F(x)F'(x)dx\nonumber\\
&-&\int_{-\infty}^{\infty}dxF'(x)\int_{-\infty}^{\infty}(1-2\chi_{(-\infty,x)}(y))K(x,y)F(y)dy
+\int_{-\infty}^{\infty}\mathrm{Ai}(x)B(x)F^2(x)dx\nonumber\\
&+&\frac{1}{4}\int_{-\infty}^{\infty}B^2(x)F(x)F'(x)dx
-\frac{1}{4}\int_{-\infty}^{\infty}dx B(x)F'(x)\int_{-\infty}^{\infty}(1-2\chi_{(-\infty,x)}(y))\mathrm{Ai}(y)F(y)dy\nonumber\\
&-&4\int_{-\infty}^{\infty}\int_{-\infty}^{\infty}K^2(x,y)F(x)F(y)dxdy+2\int_{-\infty}^{\infty}\int_{-\infty}^{\infty}K(x,y)L(x,y)F'(x)F(y)dxdy\nonumber\\
&-&2\int_{-\infty}^{\infty}\int_{-\infty}^{\infty}K(x,y)\mathrm{Ai}(x)B(y)F(x)F(y)dxdy
-\frac{1}{2}\int_{-\infty}^{\infty}\int_{-\infty}^{\infty}K(x,y)B(x)B(y)F'(x)F(y)dxdy\nonumber\\
&-&\frac{1}{4}\int_{-\infty}^{\infty}\int_{-\infty}^{\infty}L(x,y)L(y,x)F'(x)F'(y)dxdy
+\frac{1}{2}\int_{-\infty}^{\infty}\int_{-\infty}^{\infty}L(x,y)\mathrm{Ai}(y)B(x)F'(x)F(y)dxdy\nonumber\\
&+&\frac{1}{8}\int_{-\infty}^{\infty}\int_{-\infty}^{\infty}L(x,y)B(x)B(y)F'(x)F'(y)dxdy
-\frac{1}{4}\int_{-\infty}^{\infty}\int_{-\infty}^{\infty}\mathrm{Ai}(x)\mathrm{Ai}(y)B(x)B(y)F(x)F(y)dxdy\nonumber\\
&-&\frac{1}{8}\int_{-\infty}^{\infty}\int_{-\infty}^{\infty}\mathrm{Ai}(x)B(x)B^2(y)F(x)F'(y)dxdy
-\frac{1}{64}\int_{-\infty}^{\infty}\int_{-\infty}^{\infty}B^2(x)B^2(y)F'(x)F'(y)dxdy\bigg\}+O(N^{-\frac{1}{3}}).\nonumber
\eea
\end{small}
In view of $\log\: G_{N}^{(1)}(f)=\frac{1}{2}\log\det(I+T_{\mathrm{GOE}})$, we get the following theorem.
\begin{theorem}
Let $\mu_{N}^{(\mathrm{GOE})}$ and $\mathcal{V}_{N}^{(\mathrm{GOE})}$ be the mean and variance of the linear statistics\\
$\sum_{j=1}^{N}F\left(2^{\frac{1}{2}}N^{\frac{1}{6}}\left(x_{j}-\sqrt{2N}\right)\right)$. Then as $N\rightarrow\infty$,
\begin{small}
$$
\mu_{N}^{(\mathrm{GOE})}=\mu_{N}^{(\mathrm{GUE})}-\frac{1}{4}\int_{-\infty}^{\infty}L(x,x)F'(x)dx
+\frac{1}{4}\int_{-\infty}^{\infty}\mathrm{Ai}(x)B(x)F(x)dx+\frac{1}{16}\int_{-\infty}^{\infty}B^2(x)F'(x)dx+O(N^{-\frac{1}{3}}),
$$
\end{small}
\begin{small}
\bea
\mathcal{V}_{N}^{(\mathrm{GOE})}&=&2\mathcal{V}_{N}^{(\mathrm{GUE})}-\frac{1}{2}\int_{-\infty}^{\infty}L(x,x)F(x)F'(x)dx
-\frac{1}{2}\int_{-\infty}^{\infty}dxF'(x)\int_{-\infty}^{\infty}(1-2\chi_{(-\infty,x)}(y))K(x,y)F(y)dy\nonumber\\
&+&\frac{1}{2}\int_{-\infty}^{\infty}\mathrm{Ai}(x)B(x)F^2(x)dx
-\frac{1}{8}\int_{-\infty}^{\infty}dx B(x)F'(x)\int_{-\infty}^{\infty}(1-2\chi_{(-\infty,x)}(y))\mathrm{Ai}(y)F(y)dy\nonumber\\
&+&\frac{1}{8}\int_{-\infty}^{\infty}B^2(x)F(x)F'(x)dx-\frac{1}{16}\int_{-\infty}^{\infty}\int_{-\infty}^{\infty}\mathrm{Ai}(x)B(x)B^2(y)F(x)F'(y)dxdy\nonumber\\
&-&\int_{-\infty}^{\infty}\int_{-\infty}^{\infty}K(x,y)\mathrm{Ai}(x)B(y)F(x)F(y)dxdy
-\frac{1}{4}\int_{-\infty}^{\infty}\int_{-\infty}^{\infty}K(x,y)B(x)B(y)F'(x)F(y)dxdy\nonumber\\
&-&\frac{1}{8}\int_{-\infty}^{\infty}\int_{-\infty}^{\infty}L(x,y)L(y,x)F'(x)F'(y)dxdy
+\frac{1}{4}\int_{-\infty}^{\infty}\int_{-\infty}^{\infty}L(x,y)\mathrm{Ai}(y)B(x)F'(x)F(y)dxdy\nonumber\\
&+&\frac{1}{16}\int_{-\infty}^{\infty}\int_{-\infty}^{\infty}L(x,y)B(x)B(y)F'(x)F'(y)dxdy
-\frac{1}{8}\int_{-\infty}^{\infty}\int_{-\infty}^{\infty}\mathrm{Ai}(x)\mathrm{Ai}(y)B(x)B(y)F(x)F(y)dxdy\nonumber\\
&+&\int_{-\infty}^{\infty}\int_{-\infty}^{\infty}K(x,y)L(x,y)F'(x)F(y)dxdy
-\frac{1}{128}\int_{-\infty}^{\infty}\int_{-\infty}^{\infty}B^2(x)B^2(y)F'(x)F'(y)dxdy+O(N^{-\frac{1}{3}}),\nonumber
\eea
\end{small}
where $\mu_{N}^{(\mathrm{GUE})}$ and $\mathcal{V}_{N}^{(\mathrm{GUE})}$ are given by (\ref{guem}) and (\ref{guev}), respectively.
\end{theorem}

\section{Laguerre Ensembles}
\subsection{Laguerre Unitary Ensemble}
In the Laguerre case, $w(x)=x^{\alpha}\mathrm{e}^{-x},\;\alpha>-1,\; x\in \mathbb{R}^{+}$. From Lemma \ref{lemma1} we have
$$
\varphi_{j}(x)=\sqrt{\frac{\Gamma(j+1)}{\Gamma(j+\alpha+1)}}L_{j}^{(\alpha)}(x)x^{\frac{\alpha}{2}}\mathrm{e}^{-\frac{x}{2}},\;\;j=0,1,2,\ldots,
$$
where $L_{j}^{(\alpha)}(x)$ are the Laguerre polynomials of degree $j$.\\
We also have the following expansion formula,
\be\label{logdet2}
\log G_{N}^{(2)}(f)=\log\det\left(I+K_{N}^{(2)}f\right)=\mathrm{Tr}K_{N}^{(2)}f-\frac{1}{2}\mathrm{Tr}\left(K_{N}^{(2)}f\right)^{2}+\cdots.
\ee
\begin{theorem}\label{lue}
As $N\rightarrow\infty$,
$$
2^{\frac{4}{3}}N^{\frac{1}{3}}K_{N}^{(2)}\big(4N+2\alpha+2+2^{\frac{4}{3}}N^{\frac{1}{3}}x,4N+2\alpha+2+2^{\frac{4}{3}}N^{\frac{1}{3}}y\big)
=K(x,y)+O(N^{-\frac{1}{3}}),
$$
where $K(x,y)$ is the Airy kernel (\ref{airy1}).
\end{theorem}
\begin{proof}
From the asymptotic formula of Laguerre polynomials \cite{Szego} (page 201),
$$
\mathrm{e}^{-\frac{x}{2}}L_{n}^{(\alpha)}(x)=(-1)^n2^{-\al-\frac{1}{3}}n^{-\frac{1}{3}}
\mathrm{Ai}(-3^{-\frac{1}{3}}t)+O\left(n^{-1}\right),
$$
where
$$
x=4n+2\al+2-2\left(\frac{2n}{3}\right)^{\frac{1}{3}}t,
$$
we have
\be\label{asy1}
\mathrm{e}^{-\frac{x}{2}}L_{n}^{(\alpha)}(x)=(-1)^n2^{-\al-\frac{1}{3}}n^{-\frac{1}{3}}
\mathrm{Ai}\big(2^{-\frac{4}{3}}n^{-\frac{1}{3}}(x-4n-2\alpha-2)\big)+O\left(n^{-1}\right).
\ee
Using the Christoffel-Darboux formula,
$$
K_{N}^{(2)}(x,y)=-\frac{N!}{\Gamma(N+\alpha)}\frac{x^{\frac{\alpha}{2}}\mathrm{e}^{-\frac{x}{2}}L_{N}^{(\alpha)}(x)
y^{\frac{\alpha}{2}}\mathrm{e}^{-\frac{y}{2}}L_{N-1}^{(\alpha)}(y)
-y^{\frac{\alpha}{2}}\mathrm{e}^{-\frac{y}{2}}L_{N}^{(\alpha)}(y)
x^{\frac{\alpha}{2}}\mathrm{e}^{-\frac{x}{2}}L_{N-1}^{(\alpha)}(x)}{x-y}.
$$
Replacing the variables $x$ by $4N+2\alpha+2+2^{\frac{4}{3}}N^{\frac{1}{3}}x$ and $y$ by $4N+2\alpha+2+2^{\frac{4}{3}}N^{\frac{1}{3}}y$, and using (\ref{asy1}) together with Stirling's formula, we obtain the desired result after some elaborate computations.
\end{proof}
\noindent $\mathbf{Remark.}$ The above result was obtained by Forrester \cite{Forrester}, which, however, did not show the order term as well.

We now use Theorem \ref{lue} to compute (\ref{logdet2}) term by term as $N\rightarrow\infty$. We replace $f(x)$ by\\ $f\big(2^{-\frac{4}{3}}N^{-\frac{1}{3}}(x-4N-2\alpha-2)\big)$ in the following computations. The first term reads,
\bea
\mathrm{Tr}K_{N}^{(2)}f
&=&\int_{0}^{\infty}K_{N}^{(2)}(x,x)f\big(2^{-\frac{4}{3}}N^{-\frac{1}{3}}(x-4N-2\alpha-2)\big)dx\nonumber\\
&=&\int_{-2^{-\frac{4}{3}}N^{-\frac{1}{3}}(4N+2\alpha+2)}^{\infty}2^{\frac{4}{3}}N^{\frac{1}{3}}K_{N}^{(2)}\big(4N+2\alpha+2+2^{\frac{4}{3}}N^{\frac{1}{3}}x,
4N+2\alpha+2+2^{\frac{4}{3}}N^{\frac{1}{3}}x\big)f(x)dx\nonumber\\
&=&\int_{-\infty}^{\infty}K(x,x)f(x)dx+O(N^{-\frac{1}{3}}).\nonumber
\eea
The second term,
\bea
\mathrm{Tr}\left(K_{N}^{(2)}f\right)^{2}
&=&\int_{0}^{\infty}\int_{0}^{\infty}K_{N}^{(2)}(x,y)f\big(2^{-\frac{4}{3}}N^{-\frac{1}{3}}(y-4N-2\alpha-2)\big)
K_{N}^{(2)}(y,x)\nonumber\\
&&f\big(2^{-\frac{4}{3}}N^{-\frac{1}{3}}(x-4N-2\alpha-2)\big)dxdy\nonumber\\
&=&\int_{-\infty}^{\infty}\int_{-\infty}^{\infty}K^{2}(x,y)f(x)f(y)dx dy+O(N^{-\frac{1}{3}}).\nonumber
\eea
It follows that
$$
\log G_{N}^{(2)}(f)=\int_{-\infty}^{\infty}K(x,x)f(x)dx-\frac{1}{2}\int_{-\infty}^{\infty}\int_{-\infty}^{\infty}K^{2}(x,y)f(x)f(y)dx dy+\cdots+O(N^{-\frac{1}{3}}).
$$

We proceed to study the mean and variance of the scaled linear statistics $\sum_{j=1}^{N}F\big(2^{-\frac{4}{3}}N^{-\frac{1}{3}}(x-4N-2\alpha-2)\big)$. From the relation (\ref{fx}), we find
\bea
\log G_{N}^{(2)}(f)&=&-\lambda\int_{-\infty}^{\infty}K(x,x)F(x)dx\nonumber\\
&+&\frac{\lambda^{2}}{2}\bigg[\int_{-\infty}^{\infty}K(x,x)F^{2}(x)dx-\int_{-\infty}^{\infty}\int_{-\infty}^{\infty}K^{2}(x,y)F(x)F(y)dx dy\bigg]+\cdots+O(N^{-\frac{1}{3}}).\nonumber
\eea
Then we have the following theorem.
\begin{theorem}
Let $\mu_{N}^{(\mathrm{LUE})}$ and $\mathcal{V}_{N}^{(\mathrm{LUE})}$ be the mean and variance of the linear statistics\\
$\sum_{j=1}^{N}F\big(2^{-\frac{4}{3}}N^{-\frac{1}{3}}(x-4N-2\alpha-2)\big)$. Then as $N\rightarrow\infty$,
\be
\mu_{N}^{(\mathrm{LUE})}=\int_{-\infty}^{\infty}K(x,x)F(x)dx+O(N^{-\frac{1}{3}}),\label{luem}
\ee
\be
\mathcal{V}_{N}^{(\mathrm{LUE})}=\int_{-\infty}^{\infty}K(x,x)F^{2}(x)dx-\int_{-\infty}^{\infty}\int_{-\infty}^{\infty}K^{2}(x,y)F(x)F(y)dx dy+O(N^{-\frac{1}{3}}),\label{luev}
\ee
where $K(x,y)$ is the Airy kernel defined by (\ref{airy1}).
\end{theorem}
\noindent $\mathbf{Remark.}$ Comparing Sec. 2.1 and Sec. 3.1, we see that the large $N$ behavior of the MGF of a suitably scaled linear statistics in GUE are the same with a suitably scaled linear statistics in LUE. It follows that as $N\rightarrow\infty$, the mean and variance of the corresponding linear statistics are also the same in GUE and LUE.

\subsection{Laguerre Symplectic Ensemble}
For the Laguerre symplectic ensemble, $w(x)=x^{\alpha}\mathrm{e}^{-x},\;\;\alpha>0,\;\;x\in \mathbb{R}^+$. Following \cite{Min201601}, we let
$$
\psi_{2j+1}(x):=\frac{1}{\sqrt{2}}\:x\:\varphi_{2j+1}^{(\alpha-1)}(x),\;\;\;
\psi_{2j}(x):=-\frac{1}{\sqrt{2}}\:\varepsilon\varphi_{2j+1}^{(\alpha-1)}(x),\;\;\;j=0,1,2,\ldots,
$$
where $\varphi_{j}^{(\alpha-1)}(x)$ is given by
\be\label{pj}
\varphi_{j}^{(\alpha-1)}(x):=\sqrt{\frac{\Gamma(j+1)}{\Gamma(j+\alpha)}}
L_{j}^{(\alpha-1)}(x)x^{\frac{\alpha}{2}-1}\mathrm{e}^{-\frac{x}{2}},
\;\;j=0,1,2,\ldots.
\ee
It follows that $M^{(4)}$ is the direct sum of the $N$ copies of
$
\begin{pmatrix}
0&1\\
-1&0
\end{pmatrix}
$
and $(M^{(4)})^{-1}=-M^{(4)}$.
From Lemma \ref{le2}, we obtain the following results \cite{Min201601}.
\begin{theorem}
For the Laguerre symplectic ensemble,
$$
\left[G_{N}^{(4)}(f)\right]^{2}=\det(I+T_{\mathrm{LSE}}),
$$
where
\bea
T_{\mathrm{LSE}}:&=&S_{N}^{(4)}f-\frac{1}{2}S_{N}^{(4)}\varepsilon f'-\sqrt{\Big(N+\frac{1}{2}\Big)\Big(N+\frac{\alpha}{2}\Big)}
\Big(\varepsilon\varphi_{2N+1}^{(\alpha-1)}\Big)\otimes\varphi_{2N}^{(\alpha-1)}f\nonumber\\
&-&\frac{1}{2}\sqrt{\Big(N+\frac{1}{2}\Big)\Big(N+\frac{\alpha}{2}\Big)}
\Big(\varepsilon\varphi_{2N+1}^{(\alpha-1)}\Big)\otimes\Big(\varepsilon\varphi_{2N}^{(\alpha-1)}\Big)f',
\eea
and
$$
S_{N}^{(4)}(x,y)=\sum_{j=0}^{2N}x\:\varphi_{j}^{(\alpha-1)}(x)\varphi_{j}^{(\alpha-1)}(y).
$$
\end{theorem}
We also have the following expansion formula,
$$
\log\det(I+T_{\mathrm{LSE}})=\mathrm{Tr}\log(I+T_{\mathrm{LSE}})=\mathrm{Tr}\:T_{\mathrm{LSE}}-\frac{1}{2}\mathrm{Tr}\:T_{\mathrm{LSE}}^{2}+\cdots.
$$
Using the similar method in Theorem \ref{lue}, we obtain the following theorem.
\begin{theorem}\label{lse1}
As $N\rightarrow\infty$,
$$
2^{\frac{5}{3}}N^{\frac{1}{3}}\:S_{N}^{(4)}\big(8N+2\alpha+2^{\frac{5}{3}}N^{\frac{1}{3}}x, 8N+2\alpha+2^{\frac{5}{3}}N^{\frac{1}{3}}y\big)=K(x,y)+O(N^{-\frac{1}{3}}),
$$
where $K(x,y)$ is the Airy kernel (\ref{airy1}).
\end{theorem}

\begin{theorem}\label{lse2}
As $N\rightarrow\infty$, we have
$$
\varphi_{2N}^{(\alpha-1)}\big(8N+2\alpha+2^{\frac{5}{3}}N^{\frac{1}{3}}x\big)=2^{-\frac{13}{6}}N^{-\frac{5}{6}}\mathrm{Ai}(x)+O(N^{-\frac{3}{2}}),
$$
$$
\varphi_{2N+1}^{(\alpha-1)}\big(8N+2\alpha+2^{\frac{5}{3}}N^{\frac{1}{3}}x\big)=-2^{-\frac{13}{6}}N^{-\frac{5}{6}}\mathrm{Ai}(x)+O(N^{-\frac{3}{2}}),
$$
$$
\varepsilon\varphi_{2N}^{(\alpha-1)}\big(8N+2\alpha+2^{\frac{5}{3}}N^{\frac{1}{3}}x\big)=2^{-\frac{3}{2}}N^{-\frac{1}{2}}B(x)+O(N^{-\frac{7}{6}}),
$$
$$
\varepsilon\varphi_{2N+1}^{(\alpha-1)}\big(8N+2\alpha+2^{\frac{5}{3}}N^{\frac{1}{3}}x\big)=-2^{-\frac{3}{2}}N^{-\frac{1}{2}}B(x)+O(N^{-\frac{7}{6}}),
$$
where $B(x)$ is given by (\ref{bx}).
\end{theorem}

\begin{proof}
From the definition (\ref{pj}) and the asymptotics (\ref{asy1}), we have
\bea
\varphi_{2N}^{(\alpha-1)}\big(8N+2\alpha+2^{\frac{5}{3}}N^{\frac{1}{3}}x\big)&=&\sqrt{\frac{\Gamma(2N+1)}{\Gamma(2N+\alpha)}}
\big(8N+2\alpha+2^{\frac{5}{3}}N^{\frac{1}{3}}x\big)^{\frac{\alpha}{2}-1}\big(2^{-\alpha+\frac{1}{3}}N^{-\frac{1}{3}}\mathrm{Ai}(x)+O(N^{-1})\big)\nonumber\\
&=&2^{-\frac{13}{6}}N^{-\frac{5}{6}}\mathrm{Ai}(x)+O(N^{-\frac{3}{2}}),\nonumber
\eea
where we have made use of the formula \cite{Abramowitz} (page 257)
$$
\frac{\Gamma(n+a)}{\Gamma(n+b)}=n^{a-b}(1+O(n^{-1})),\;\;n\rightarrow\infty.
$$

It follows that
\bea
&&\varepsilon\varphi_{2N}^{(\alpha-1)}\big(8N+2\alpha+2^{\frac{5}{3}}N^{\frac{1}{3}}x\big)\nonumber\\
&=&\frac{1}{2}\left(\int_{0}^{8N+2\alpha+2^{\frac{5}{3}}N^{\frac{1}{3}}x}\varphi_{2N}^{(\alpha-1)}(y)dy-\int_{8N+2\alpha+2^{\frac{5}{3}}N^{\frac{1}{3}}x}
^{\infty}\varphi_{2N}^{(\alpha-1)}(y)dy\right)\nonumber\\
&=&2^{\frac{2}{3}}N^{\frac{1}{3}}\left(\int_{-2^{-\frac{2}{3}}N^{-\frac{1}{3}}(4N+\alpha)}^{x}\varphi_{2N}^{(\alpha-1)}(8N+2\alpha+2^{\frac{5}{3}}
N^{\frac{1}{3}}y)dy-\int_{x}^{\infty}\varphi_{2N}^{(\alpha-1)}(8N+2\alpha+2^{\frac{5}{3}}N^{\frac{1}{3}}y)dy\right)\nonumber\\
&=&2^{\frac{2}{3}}N^{\frac{1}{3}}\left(\int_{-\infty}^{x}\left(2^{-\frac{13}{6}}N^{-\frac{5}{6}}\mathrm{Ai}(y)+O(N^{-\frac{3}{2}})\right)dy
-\int_{x}^{\infty}\left(2^{-\frac{13}{6}}N^{-\frac{5}{6}}\mathrm{Ai}(y)+O(N^{-\frac{3}{2}})\right)dy\right)\nonumber\\
&=&2^{-\frac{3}{2}}N^{-\frac{1}{2}}B(x)+O(N^{-\frac{7}{6}}).\nonumber
\eea
Similarly, we obtain
$$
\varphi_{2N+1}^{(\alpha-1)}\big(8N+2\alpha+2^{\frac{5}{3}}N^{\frac{1}{3}}x\big)=-2^{-\frac{13}{6}}N^{-\frac{5}{6}}\mathrm{Ai}(x)+O(N^{-\frac{3}{2}})
$$
and
$$
\varepsilon\varphi_{2N+1}^{(\alpha-1)}\big(8N+2\alpha+2^{\frac{5}{3}}N^{\frac{1}{3}}x\big)=-2^{-\frac{3}{2}}N^{-\frac{1}{2}}B(x)+O(N^{-\frac{7}{6}}).
$$
The proof is complete.
\end{proof}

Now we use Theorem \ref{lse1} and \ref{lse2} to compute $\mathrm{Tr}\:T_{\mathrm{LSE}}$ and $\mathrm{Tr}\:T_{\mathrm{LSE}}^{2}$ as $N\rightarrow\infty$. We change $f(x)$ to $f(2^{-\frac{5}{3}}N^{-\frac{1}{3}}(x-8N-2\alpha))$ in the following computations. Firstly we have
\bea
\mathrm{Tr}\:T_{\mathrm{LSE}}&=&\mathrm{Tr}\:S_{N}^{(4)}f-\mathrm{Tr}\:\frac{1}{2}S_{N}^{(4)}\varepsilon f'-\mathrm{Tr}\:\sqrt{\Big(N+\frac{1}{2}\Big)\Big(N+\frac{\alpha}{2}\Big)}
\Big(\varepsilon\varphi_{2N+1}^{(\alpha-1)}\Big)\otimes\varphi_{2N}^{(\alpha-1)}f\nonumber\\
&-&\mathrm{Tr}\:\frac{1}{2}\sqrt{\Big(N+\frac{1}{2}\Big)\Big(N+\frac{\alpha}{2}\Big)}
\Big(\varepsilon\varphi_{2N+1}^{(\alpha-1)}\Big)\otimes\Big(\varepsilon\varphi_{2N}^{(\alpha-1)}\Big)f'.\nonumber
\eea
The first term reads,
\bea
\mathrm{Tr}\:S_{N}^{(4)}f
&=&\int_{0}^{\infty}S_{N}^{(4)}(x,x)f(2^{-\frac{5}{3}}N^{-\frac{1}{3}}(x-8N-2\alpha))dx\nonumber\\
&=&\int_{-2^{-\frac{2}{3}}N^{-\frac{1}{3}}(4N+\alpha)}^{\infty}2^{\frac{5}{3}}N^{\frac{1}{3}}\:S_{N}^{(4)}\big(8N+2\alpha+2^{\frac{5}{3}}N^{\frac{1}{3}}x, 8N+2\alpha+2^{\frac{5}{3}}N^{\frac{1}{3}}x\big)f(x)dx\nonumber\\
&=&\int_{-\infty}^{\infty}K(x,x)f(x)dx+O(N^{-\frac{1}{3}}).\nonumber
\eea
The second term,
\bea
\mathrm{Tr}\:\frac{1}{2}S_{N}^{(4)}\varepsilon f'
&=&2^{-\frac{8}{3}}N^{-\frac{1}{3}}\int_{0}^{\infty}\int_{0}^{\infty}S_{N}^{(4)}(x,y)\varepsilon(y,x)f'(2^{-\frac{5}{3}}N^{-\frac{1}{3}}(x-8N-2\alpha))dxdy
\nonumber\\
&=&2^{-\frac{11}{3}}N^{-\frac{1}{3}}\int_{0}^{\infty}\left(\int_{x}^{\infty}S_{N}^{(4)}(x,y)dy-\int_{0}^{x}S_{N}^{(4)}(x,y)dy\right)
f'(2^{-\frac{5}{3}}N^{-\frac{1}{3}}(x-8N-2\alpha))dx.\nonumber
\eea
Let
$$
x=8N+2\alpha+2^{\frac{5}{3}}N^{\frac{1}{3}}u,\;\;y=8N+2\alpha+2^{\frac{5}{3}}N^{\frac{1}{3}}v,
$$
then
\bea
\mathrm{Tr}\:\frac{1}{2}S_{N}^{(4)}\varepsilon f'
&=&\frac{1}{4}\int_{-\infty}^{\infty}\left(\int_{u}^{\infty}K(u,v)dv-\int_{-\infty}^{u}K(u,v)dv\right)f'(u)du+O(N^{-\frac{1}{3}})\nonumber\\
&=&\frac{1}{4}\int_{-\infty}^{\infty}L(x,x)f'(x)dx+O(N^{-\frac{1}{3}}).\nonumber
\eea
The third term,
\bea
&&\mathrm{Tr}\:\sqrt{\Big(N+\frac{1}{2}\Big)\Big(N+\frac{\alpha}{2}\Big)}
\Big(\varepsilon\varphi_{2N+1}^{(\alpha-1)}\Big)\otimes\varphi_{2N}^{(\alpha-1)}f\nonumber\\
&=&\sqrt{\Big(N+\frac{1}{2}\Big)\Big(N+\frac{\alpha}{2}\Big)}\int_{0}^{\infty}\varepsilon\varphi_{2N+1}^{(\alpha-1)}(x)\varphi_{2N}^{(\alpha-1)}(x)
f(2^{-\frac{5}{3}}N^{-\frac{1}{3}}(x-8N-2\alpha))dx\nonumber\\
&=&2^{\frac{5}{3}}N^{\frac{1}{3}}\sqrt{\Big(N+\frac{1}{2}\Big)\Big(N+\frac{\alpha}{2}\Big)}
\int_{-2^{-\frac{2}{3}}N^{-\frac{1}{3}}(4N+\alpha)}^{\infty}\varepsilon\varphi_{2N+1}^{(\alpha-1)}(8N+2\alpha+2^{\frac{5}{3}}N^{\frac{1}{3}}x)\nonumber\\
&&\varphi_{2N}^{(\alpha-1)}(8N+2\alpha+2^{\frac{5}{3}}N^{\frac{1}{3}}x)f(x)dx\nonumber\\
&=&-\frac{1}{4}\int_{-\infty}^{\infty}\mathrm{Ai}(x)B(x)f(x)dx+O(N^{-\frac{2}{3}}).\nonumber
\eea
The fourth term,
\bea
&&\mathrm{Tr}\:\frac{1}{2}\sqrt{\Big(N+\frac{1}{2}\Big)\Big(N+\frac{\alpha}{2}\Big)}
\Big(\varepsilon\varphi_{2N+1}^{(\alpha-1)}\Big)\otimes\Big(\varepsilon\varphi_{2N}^{(\alpha-1)}\Big)f'\nonumber\\
&=&2^{-\frac{8}{3}}N^{-\frac{1}{3}}\sqrt{\Big(N+\frac{1}{2}\Big)\Big(N+\frac{\alpha}{2}\Big)}\int_{0}^{\infty}\varepsilon\varphi_{2N+1}^{(\alpha-1)}(x)
\varepsilon\varphi_{2N}^{(\alpha-1)}(x)f'(2^{-\frac{5}{3}}N^{-\frac{1}{3}}(x-8N-2\alpha))dx\nonumber\\
&=&-\frac{1}{16}\int_{-\infty}^{\infty}B^2(x)f'(x)dx+O(N^{-\frac{2}{3}}).\nonumber
\eea
Hence,
\bea\label{tr1}
\mathrm{Tr}\:T_{\mathrm{LSE}}&=&\int_{-\infty}^{\infty}K(x,x)f(x)dx-\frac{1}{4}\int_{-\infty}^{\infty}L(x,x)f'(x)dx
+\frac{1}{4}\int_{-\infty}^{\infty}\mathrm{Ai}(x)B(x)f(x)dx\nonumber\\
&+&\frac{1}{16}\int_{-\infty}^{\infty}B^2(x)f'(x)dx+O(N^{-\frac{1}{3}}).
\eea
Similarly we obtain the result for $\mathrm{Tr}\:T_{\mathrm{LSE}}^{2}$ after some tedious computations,
\bea\label{tr2}
&&\mathrm{Tr}\:T_{\mathrm{LSE}}^{2}\nonumber\\
&=&\int_{-\infty}^{\infty}\int_{-\infty}^{\infty}K^2(x,y)f(x)f(y)dxdy-\frac{1}{2}\int_{-\infty}^{\infty}\int_{-\infty}^{\infty}K(x,y)L(x,y)f'(x)f(y)dxdy
\nonumber\\
&+&\frac{1}{2}\int_{-\infty}^{\infty}\int_{-\infty}^{\infty}K(x,y)\mathrm{Ai}(x)B(y)f(x)f(y)dxdy
+\frac{1}{8}\int_{-\infty}^{\infty}\int_{-\infty}^{\infty}K(x,y)B(x)B(y)f'(x)f(y)dxdy\nonumber\\
&+&\frac{1}{16}\int_{-\infty}^{\infty}\int_{-\infty}^{\infty}L(x,y)L(y,x)f'(x)f'(y)dxdy
-\frac{1}{8}\int_{-\infty}^{\infty}\int_{-\infty}^{\infty}L(x,y)\mathrm{Ai}(y)B(x)f'(x)f(y)dxdy\nonumber\\
&-&\frac{1}{32}\int_{-\infty}^{\infty}\int_{-\infty}^{\infty}L(x,y)B(x)B(y)f'(x)f'(y)dxdy+
\frac{1}{16}\int_{-\infty}^{\infty}\int_{-\infty}^{\infty}\mathrm{Ai}(x)\mathrm{Ai}(y)B(x)B(y)f(x)f(y)dxdy\nonumber\\
&+&\frac{1}{32}\int_{-\infty}^{\infty}\int_{-\infty}^{\infty}\mathrm{Ai}(x)B(x)B^2(y)f(x)f'(y)dxdy
+\frac{1}{256}\int_{-\infty}^{\infty}\int_{-\infty}^{\infty}B^2(x)B^2(y)f'(x)f'(y)dxdy+O(N^{-\frac{1}{3}}).\nonumber\\
\eea
From the above, we find that as $N\rightarrow\infty$,
$$
\mathrm{Tr}\:T_{\mathrm{LSE}}=\mathrm{Tr}\:T_{\mathrm{GSE}},\;\;\;\;\mathrm{Tr}\:T_{\mathrm{LSE}}^{2}=\mathrm{Tr}\:T_{\mathrm{GSE}}^{2}.
$$
It follows that as $N\rightarrow\infty$,
$$
\log\det(I+T_{\mathrm{LSE}})=\log\det(I+T_{\mathrm{GSE}}).
$$
Then we have the following theorem.
\begin{theorem}
Let $\mu_{N}^{(\mathrm{LSE})}$ and $\mathcal{V}_{N}^{(\mathrm{LSE})}$ be the mean and variance of the linear statistics\\
$\sum_{j=1}^{N}F\left(2^{-\frac{5}{3}}N^{-\frac{1}{3}}(x_{j}-8N-2\alpha)\right)$. Then as $N\rightarrow\infty$,
\begin{small}
$$
\mu_{N}^{(\mathrm{LSE})}=\frac{1}{2}\mu_{N}^{(\mathrm{LUE})}-\frac{1}{8}\int_{-\infty}^{\infty}L(x,x)F'(x)dx
+\frac{1}{8}\int_{-\infty}^{\infty}\mathrm{Ai}(x)B(x)F(x)dx+\frac{1}{32}\int_{-\infty}^{\infty}B^2(x)F'(x)dx+O(N^{-\frac{1}{3}}),
$$
\end{small}
\begin{small}
\bea
\mathcal{V}_{N}^{(\mathrm{LSE})}&=&\frac{1}{2}\mathcal{V}_{N}^{(\mathrm{LUE})}-\frac{1}{4}\int_{-\infty}^{\infty}L(x,x)F(x)F'(x)dx
+\frac{1}{8}\int_{-\infty}^{\infty}\mathrm{Ai}(x)B(x)F^2(x)dx+\frac{1}{16}\int_{-\infty}^{\infty}B^2(x)F(x)F'(x)dx\nonumber\\
&+&\frac{1}{4}\int_{-\infty}^{\infty}\int_{-\infty}^{\infty}K(x,y)L(x,y)F'(x)F(y)dxdy
-\frac{1}{4}\int_{-\infty}^{\infty}\int_{-\infty}^{\infty}K(x,y)\mathrm{Ai}(x)B(y)F(x)F(y)dxdy\nonumber\\
&-&\frac{1}{16}\int_{-\infty}^{\infty}\int_{-\infty}^{\infty}K(x,y)B(x)B(y)F'(x)F(y)dxdy
-\frac{1}{32}\int_{-\infty}^{\infty}\int_{-\infty}^{\infty}L(x,y)L(y,x)F'(x)F'(y)dxdy\nonumber\\
&+&\frac{1}{16}\int_{-\infty}^{\infty}\int_{-\infty}^{\infty}L(x,y)\mathrm{Ai}(y)B(x)F'(x)F(y)dxdy
+\frac{1}{64}\int_{-\infty}^{\infty}\int_{-\infty}^{\infty}L(x,y)B(x)B(y)F'(x)F'(y)dxdy\nonumber\\
&-&\frac{1}{32}\int_{-\infty}^{\infty}\int_{-\infty}^{\infty}\mathrm{Ai}(x)\mathrm{Ai}(y)B(x)B(y)F(x)F(y)dxdy
-\frac{1}{64}\int_{-\infty}^{\infty}\int_{-\infty}^{\infty}\mathrm{Ai}(x)B(x)B^2(y)F(x)F'(y)dxdy\nonumber\\
&-&\frac{1}{512}\int_{-\infty}^{\infty}\int_{-\infty}^{\infty}B^2(x)B^2(y)F'(x)F'(y)dxdy+O(N^{-\frac{1}{3}}),\nonumber
\eea
\end{small}
where $\mu_{N}^{(\mathrm{LUE})}$ and $\mathcal{V}_{N}^{(\mathrm{LUE})}$ are given by (\ref{luem}) and (\ref{luev}), respectively.
\end{theorem}
\noindent $\mathbf{Remark.}$ We find that the large $N$ behavior of the MGF of a suitably scaled linear statistics in LSE are the same with a suitably scaled linear statistics in GSE. It follows that as $N\rightarrow\infty$, the mean and variance of the corresponding linear statistics are also the same in LSE and GSE.

\subsection{Laguerre Orthogonal Ensemble}
In this subsection, $w(x)$ is taken to be the {\it square root} of the Laguerre weight, namely,
$$
w(x)=x^{\frac{\alpha}{2}}\mathrm{e}^{-\frac{x}{2}},\;\;\alpha>-2,\;x\in \mathbb{R}^+,
$$
and $N$ is even. Following \cite{Min201601}, we let
$$
\psi_{2n+1}(x):=\frac{d}{dx}(x\varphi_{2n}^{(\alpha+1)}(x)),\;\;\psi_{2n}(x):=\varphi_{2n}^{(\alpha+1)}(x),\;\;n=0,1,2,\ldots,
$$
where $\varphi_{j}^{(\alpha+1)}(x)$ is given by
$$
\varphi_{j}^{(\alpha+1)}(x):=\sqrt{\frac{\Gamma(j+1)}{\Gamma(j+\alpha+2)}}L_{j}^{(\alpha+1)}(x)x^{\frac{\alpha}{2}}\mathrm{e}^{-\frac{x}{2}},\;\;j=0,1,2,\ldots.
$$
Note that the definition of $\varphi_{j}^{(\alpha+1)}(x)$ coincides with the LSE case if we replace $\alpha$ with $\alpha+2$ there.
It follows that $M^{(1)}$ is the direct sum of the $\frac{N}{2}$ copies of
$
\begin{pmatrix}
0&1\\
-1&0
\end{pmatrix}
$
and $(M^{(1)})^{-1}=-M^{(1)}$.
From Lemma \ref{le3}, we obtain the following result \cite{Min201601}.
\begin{theorem}
For the Laguerre orthogonal ensemble,
$$
\left[G_{N}^{(1)}(f)\right]^{2}=\det(I+T_{\mathrm{LOE}}),
$$
where
\bea
T_{\mathrm{LOE}}:&=&S_{N}^{(1)}(f^2+2f)-S_{N}^{(1)}\varepsilon f'-S_{N}^{(1)}f\varepsilon f'-\frac{1}{2}\sqrt{N(N+\alpha+1)}
\big(\varepsilon\varphi_{N}^{(\alpha+1)}\big)\otimes\varphi_{N-1}^{(\alpha+1)}(f^2+2f)\nonumber\\
&-&\frac{1}{2}\sqrt{N(N+\alpha+1)}\big(\varepsilon\varphi_{N}^{(\alpha+1)}\big)\otimes\big(\varepsilon\varphi_{N-1}^{(\alpha+1)}\big)f'
+\frac{1}{2}\sqrt{N(N+\alpha+1)}\big(\varepsilon\varphi_{N}^{(\alpha+1)}\big)\otimes\varphi_{N-1}^{(\alpha+1)}f\varepsilon f',\nonumber
\eea
and $S_{N}^{(1)}$ is an integral operator with kernel
$$
S_{N}^{(1)}(x,y)=\sum_{j=0}^{N-1}x\:\varphi_{j}^{(\alpha+1)}(x)\varphi_{j}^{(\alpha+1)}(y).
$$
\end{theorem}
We also have the following expansion formula,
$$
\log\det(I+T_{\mathrm{LOE}})=\mathrm{Tr}\log(I+T_{\mathrm{LOE}})=\mathrm{Tr}\:T_{\mathrm{LOE}}-\frac{1}{2}\mathrm{Tr}\:T_{\mathrm{LOE}}^{2}+\cdots.
$$
Similarly as previous subsection, we have the following theorems.
\begin{theorem}\label{loe1}
As $N\rightarrow\infty$,
$$
2^{\frac{4}{3}}N^{\frac{1}{3}}\:S_{N}^{(1)}\big(4N+2\alpha+4+2^{\frac{4}{3}}N^{\frac{1}{3}}x, 4N+2\alpha+4+2^{\frac{4}{3}}N^{\frac{1}{3}}y\big)=K(x,y)+O(N^{-\frac{1}{3}}),
$$
where $K(x,y)$ is the Airy kernel (\ref{airy1}).
\end{theorem}

\begin{theorem}\label{loe2}
As $N\rightarrow\infty$, we have
$$
\varphi_{N}^{(\alpha+1)}\big(4N+2\alpha+4+2^{\frac{4}{3}}N^{\frac{1}{3}}x\big)=2^{-\frac{4}{3}}N^{-\frac{5}{6}}\mathrm{Ai}(x)+O(N^{-\frac{3}{2}}),
$$
$$
\varphi_{N-1}^{(\alpha+1)}\big(4N+2\alpha+4+2^{\frac{4}{3}}N^{\frac{1}{3}}x\big)=-2^{-\frac{4}{3}}N^{-\frac{5}{6}}\mathrm{Ai}(x)+O(N^{-\frac{3}{2}}),
$$
$$
\varepsilon\varphi_{N}^{(\alpha+1)}\big(4N+2\alpha+4+2^{\frac{4}{3}}N^{\frac{1}{3}}x\big)=2^{-1}N^{-\frac{1}{2}}B(x)+O(N^{-\frac{7}{6}}),
$$
$$
\varepsilon\varphi_{N-1}^{(\alpha+1)}\big(4N+2\alpha+4+2^{\frac{4}{3}}N^{\frac{1}{3}}x\big)=-2^{-1}N^{-\frac{1}{2}}B(x)+O(N^{-\frac{7}{6}}).
$$
\end{theorem}

Now we change $f(x)$ to $f(2^{-\frac{4}{3}}N^{-\frac{1}{3}}(x-4N-2\alpha-4))$ and use Theorem \ref{loe1} and \ref{loe2} to compute $\mathrm{Tr}\:T_{\mathrm{LOE}}$ and $\mathrm{Tr}\:T_{\mathrm{LOE}}^{2}$. We find that as $N\rightarrow\infty$,
$$
\mathrm{Tr}\:T_{\mathrm{LOE}}=\mathrm{Tr}\:T_{\mathrm{GOE}},\;\;\;\;\mathrm{Tr}\:T_{\mathrm{LOE}}^2=\mathrm{Tr}\:T_{\mathrm{GOE}}^2.
$$
It follows that as $N\rightarrow\infty$,
$$
\log\det(I+T_{\mathrm{LOE}})=\log\det(I+T_{\mathrm{GOE}}).
$$
Then we have the following theorem.
\begin{theorem}
Denoting by $\mu_{N}^{(\mathrm{LOE})}$ and $\mathcal{V}_{N}^{(\mathrm{LOE})}$ the mean and variance of the scaled linear statistics\\
$\sum_{j=1}^{N}F\left(2^{-\frac{4}{3}}N^{-\frac{1}{3}}(x-4N-2\alpha-4)\right)$, then as $N\rightarrow\infty$,
\begin{small}
$$
\mu_{N}^{(\mathrm{LOE})}=\mu_{N}^{(\mathrm{LUE})}-\frac{1}{4}\int_{-\infty}^{\infty}L(x,x)F'(x)dx
+\frac{1}{4}\int_{-\infty}^{\infty}\mathrm{Ai}(x)B(x)F(x)dx+\frac{1}{16}\int_{-\infty}^{\infty}B^2(x)F'(x)dx+O(N^{-\frac{1}{3}}),
$$
\end{small}
\begin{small}
\bea
\mathcal{V}_{N}^{(\mathrm{LOE})}&=&2\mathcal{V}_{N}^{(\mathrm{LUE})}-\frac{1}{2}\int_{-\infty}^{\infty}L(x,x)F(x)F'(x)dx
-\frac{1}{2}\int_{-\infty}^{\infty}dxF'(x)\int_{-\infty}^{\infty}(1-2\chi_{(-\infty,x)}(y))K(x,y)F(y)dy\nonumber\\
&+&\frac{1}{2}\int_{-\infty}^{\infty}\mathrm{Ai}(x)B(x)F^2(x)dx
-\frac{1}{8}\int_{-\infty}^{\infty}dx B(x)F'(x)\int_{-\infty}^{\infty}(1-2\chi_{(-\infty,x)}(y))\mathrm{Ai}(y)F(y)dy\nonumber\\
&+&\frac{1}{8}\int_{-\infty}^{\infty}B^2(x)F(x)F'(x)dx-\frac{1}{16}\int_{-\infty}^{\infty}\int_{-\infty}^{\infty}\mathrm{Ai}(x)B(x)B^2(y)F(x)F'(y)dxdy\nonumber\\
&-&\int_{-\infty}^{\infty}\int_{-\infty}^{\infty}K(x,y)\mathrm{Ai}(x)B(y)F(x)F(y)dxdy
-\frac{1}{4}\int_{-\infty}^{\infty}\int_{-\infty}^{\infty}K(x,y)B(x)B(y)F'(x)F(y)dxdy\nonumber\\
&-&\frac{1}{8}\int_{-\infty}^{\infty}\int_{-\infty}^{\infty}L(x,y)L(y,x)F'(x)F'(y)dxdy
+\frac{1}{4}\int_{-\infty}^{\infty}\int_{-\infty}^{\infty}L(x,y)\mathrm{Ai}(y)B(x)F'(x)F(y)dxdy\nonumber\\
&+&\frac{1}{16}\int_{-\infty}^{\infty}\int_{-\infty}^{\infty}L(x,y)B(x)B(y)F'(x)F'(y)dxdy
-\frac{1}{8}\int_{-\infty}^{\infty}\int_{-\infty}^{\infty}\mathrm{Ai}(x)\mathrm{Ai}(y)B(x)B(y)F(x)F(y)dxdy\nonumber\\
&+&\int_{-\infty}^{\infty}\int_{-\infty}^{\infty}K(x,y)L(x,y)F'(x)F(y)dxdy
-\frac{1}{128}\int_{-\infty}^{\infty}\int_{-\infty}^{\infty}B^2(x)B^2(y)F'(x)F'(y)dxdy+O(N^{-\frac{1}{3}}),\nonumber
\eea
\end{small}
where $\mu_{N}^{(\mathrm{LUE})}$ and $\mathcal{V}_{N}^{(\mathrm{LUE})}$ are given by (\ref{luem}) and (\ref{luev}), respectively.
\end{theorem}
\noindent $\mathbf{Remark.}$ We find that the large $N$ behavior of the MGF of a suitably scaled linear statistics in LOE are the same with a suitably scaled linear statistics in GOE. It follows that as $N\rightarrow\infty$, the mean and variance of the corresponding linear statistics are also the same in LOE and GOE.

\section{Gaussian Unitary Ensemble Continued}
For the Gaussian unitary ensemble, if we change $f(x)$ to $f\left(x-\sqrt{2N}\right)$, we can gain a better insight into the mean and variance of the corresponding linear statistics by using the result of Basor and Widom \cite{Basor1999}. We see that as $N\rightarrow\infty$,
$$
\det(I+K_{N}^{(2)}f)=\det(I+K\tilde{f}),
$$
where
$$
\tilde{f}(x):=f\left(\frac{x}{2^{\frac{1}{2}}N^{\frac{1}{6}}}\right).
$$
This is because, as $N\rightarrow\infty$,
\bea
\mathrm{Tr}\:K_{N}^{(2)}f
&=&\int_{-\infty}^{\infty}K_{N}^{(2)}(x,x)f\left(x-\sqrt{2N}\right)dx\nonumber\\
&=&\int_{-\infty}^{\infty}2^{-\frac{1}{2}}N^{-\frac{1}{6}}K_{N}^{(2)}
\left(\sqrt{2N}+2^{-\frac{1}{2}}N^{-\frac{1}{6}}x,\sqrt{2N}+2^{-\frac{1}{2}}N^{-\frac{1}{6}}x\right)f\left(\frac{x}{2^{\frac{1}{2}}N^{\frac{1}{6}}}\right)dx
\nonumber\\
&=&\int_{-\infty}^{\infty}K(x,x)f\left(\frac{x}{2^{\frac{1}{2}}N^{\frac{1}{6}}}\right)dx\nonumber\\
&=&\mathrm{Tr}\:K\tilde{f},\nonumber
\eea
\bea
\mathrm{Tr}\left(K_{N}^{(2)}f\right)^{2}
&=&\int_{-\infty}^{\infty}\int_{-\infty}^{\infty}K_{N}^{(2)}(x,y)f\left(y-\sqrt{2N}\right)K_{N}^{(2)}(y,x)
f\left(x-\sqrt{2N}\right)dx dy\nonumber\\
&=&\int_{-\infty}^{\infty}\int_{-\infty}^{\infty}2^{-\frac{1}{2}}N^{-\frac{1}{6}}K_{N}^{(2)}
\left(\sqrt{2N}+2^{-\frac{1}{2}}N^{-\frac{1}{6}}x,\sqrt{2N}+2^{-\frac{1}{2}}N^{-\frac{1}{6}}y\right)f\left(\frac{y}{2^{\frac{1}{2}}N^{\frac{1}{6}}}\right)
\nonumber\\
&\cdot&2^{-\frac{1}{2}}N^{-\frac{1}{6}}K_{N}^{(2)}\left(\sqrt{2N}+2^{-\frac{1}{2}}N^{-\frac{1}{6}}y,\sqrt{2N}+2^{-\frac{1}{2}}N^{-\frac{1}{6}}x\right)
f\left(\frac{x}{2^{\frac{1}{2}}N^{\frac{1}{6}}}\right)dx dy\nonumber\\
&=&\int_{-\infty}^{\infty}\int_{-\infty}^{\infty}K^{2}(x,y)f\left(\frac{x}{2^{\frac{1}{2}}N^{\frac{1}{6}}}\right)
f\left(\frac{y}{2^{\frac{1}{2}}N^{\frac{1}{6}}}\right)dx dy\nonumber\\
&=&\mathrm{Tr}\left(K\tilde{f}\right)^{2}\nonumber
\eea
and so on.

We now introduce the result of Basor and Widom as the following lemma \cite{Basor1999}.
\begin{lemma}\label{le4}
Let
$$
\hat{f}(x):=f\left(\frac{x}{\gamma}\right),
$$
then as $\gamma\rightarrow\infty$,
$$
\log\det(I+K\hat{f})=c_{1}\gamma^{\frac{3}{2}}+c_{2}+o(1),
$$
where
$$
c_{1}=\frac{1}{\pi}\int_{0}^{\infty}\sqrt{x}\log(1+f(-x))dx,
$$
$$
c_{2}=\frac{1}{2}\int_{0}^{\infty}x\: G^2(x)dx
$$
and
$$
G(x)=\frac{1}{2\pi}\int_{-\infty}^{\infty}\mathrm{e}^{ixy}\log(1+f(-y^2))dy.
$$
\end{lemma}
From Lemma \ref{le4} and noting that $\gamma=2^{\frac{1}{2}}N^{\frac{1}{6}}$ in our problem, we obtain as $N\rightarrow\infty$,
\be\label{bw}
\log G_{N}^{(2)}(f)=\frac{2^{\frac{3}{4}}N^{\frac{1}{4}}}{\pi}\int_{0}^{\infty}\sqrt{x}\log(1+f(-x))dx+\frac{1}{2}\int_{0}^{\infty}x\: G^2(x)dx+o(1),
\ee
Substituting $f(x)={\rm e}^{-\lambda F(x)}-1$ into (\ref{bw}), we find
\bea
\log G_{N}^{(2)}(f)&=&-\frac{2^{\frac{3}{4}}N^{\frac{1}{4}}\lambda}{\pi}\int_{0}^{\infty}\sqrt{x}F(-x)dx+\frac{\lambda^{2}}{8\pi^{2}}\int_{0}^{\infty}
x\left(\int_{-\infty}^{\infty}\mathrm{e}^{ixy}F(-y^2)dy\right)^2dx+o(1)\nonumber\\
&=&-\frac{2^{\frac{3}{4}}N^{\frac{1}{4}}\lambda}{\pi}\int_{0}^{\infty}\sqrt{x}F(-x)dx+\frac{\lambda^{2}}{2\pi^{2}}\int_{0}^{\infty}
x\left(\int_{0}^{\infty}\cos(xy)F(-y^2)dy\right)^2dx+o(1).\nonumber
\eea
Hence we have the following theorem.
\begin{theorem}\label{mv}
As $N\rightarrow\infty$, the mean and variance of the linear statistics $\sum_{j=1}^{N}F\left(x_{j}-\sqrt{2N}\right)$ in Gaussian unitary ensemble are given by
$$
\mu=\frac{2^{\frac{3}{4}}N^{\frac{1}{4}}}{\pi}\int_{0}^{\infty}\sqrt{x}F(-x)dx+o(1)
$$
and
$$
\mathcal{V}=\frac{1}{\pi^{2}}\int_{0}^{\infty}x\left(\int_{0}^{\infty}\cos(xy)F(-y^2)dy\right)^2dx+o(1)
$$
respectively.
\end{theorem}

At the end of this section, we use another method, the coulomb fluid approach, to prove the above theorem. We state an important lemma \cite{Basor2001}.
\begin{lemma}
As $N\rightarrow\infty$,
$$
\mathbb{E}\left({\rm e}^{-\lambda\:\sum_{j=1}^{N}F(x_j)}\right)\sim \exp(-S_{1}-S_{2}),
$$
where
$$
S_{1}=\frac{\lambda^2}{4\pi^2}\int_{a}^{b}\int_{a}^{b}\frac{F(x)F(y)}{\sqrt{(b-x)(x-a)}}\frac{\partial}{\partial y}\left(\frac{\sqrt{(b-y)(y-a)}}{x-y}\right)dxdy,
$$
$$
S_{2}=\lambda\int_{a}^{b}\sigma(x)F(x)dx,
$$
and $\sigma(x)$ is the equilibrium density of the eigenvalues (particles) supported on the interval $(a,b)$.
\end{lemma}
For the Gaussian unitary ensemble, it is known that \cite{Basor2001}
$$
\sigma(x)=\frac{\sqrt{b^2-x^2}}{\pi},\;\; b=-a=\sqrt{2N}.
$$
In our case, we replace $F(x)$ by $F\left(x-\sqrt{2N}\right)$. We have as $N\rightarrow\infty$,
$$
\mathbb{E}\left({\rm e}^{-\lambda\:\sum_{j=1}^{N}F\left(x_j-\sqrt{2N}\right)}\right)\sim \exp(-S_{1}-S_{2}),
$$
where
$$
S_{1}=\frac{\lambda^2}{4\pi^2}\int_{-\sqrt{2N}}^{\sqrt{2N}}\int_{-\sqrt{2N}}^{\sqrt{2N}}\frac{F(x-\sqrt{2N})F(y-\sqrt{2N})}{\sqrt{2N-x^2}}
\frac{\partial}{\partial y}\left(\frac{\sqrt{2N-y^2}}{x-y}\right)dxdy,
$$
$$
S_{2}=\frac{\lambda}{\pi}\int_{-\sqrt{2N}}^{\sqrt{2N}}\sqrt{2N-x^2}F\left(x-\sqrt{2N}\right)dx.
$$
Firstly, we compute $S_{1}$.
Let $x=\sqrt{2N}-u,\; y=\sqrt{2N}-v$,
$$
S_{1}
=-\frac{\lambda^2}{4\pi^2}\int_{0}^{2\sqrt{2N}}\int_{0}^{2\sqrt{2N}}\frac{F(-u)F(-v)}{\sqrt{u\left(2\sqrt{2N}-u\right)}}
\frac{\partial}{\partial v}\left(\frac{\sqrt{v\left(2\sqrt{2N}-v\right)}}{v-u}\right)dudv.
$$
As $N\rightarrow\infty$,
$$
S_{1}
\sim-\frac{\lambda^2}{4\pi^2}\int_{0}^{\infty}\int_{0}^{\infty}\frac{F(-u)F(-v)}{\sqrt{u}}
\frac{\partial}{\partial v}\left(\frac{\sqrt{v}}{v-u}\right)dudv.
$$
After the change of variables $u=s^2, v=t^2$, we find
\bea
S_{1}&\sim&\frac{\lambda^2}{2\pi^2}\int_{0}^{\infty}\int_{0}^{\infty}F(-s^2)F(-t^2)\frac{s^2+t^2}{(s^2-t^2)^2}dsdt\nonumber\\
&=&\frac{\lambda^2}{8\pi^2}\int_{-\infty}^{\infty}\int_{-\infty}^{\infty}F(-s^2)F(-t^2)\frac{s^2+t^2}{(s^2-t^2)^2}dsdt\nonumber\\
&=&\frac{\lambda^2}{8\pi^2}\int_{-\infty}^{\infty}\int_{-\infty}^{\infty}F(-s^2)F(-t^2)\frac{s^2-2st+t^2}{(s^2-t^2)^2}dsdt\nonumber\\
&=&\frac{\lambda^2}{8\pi^2}\int_{-\infty}^{\infty}\int_{-\infty}^{\infty}F(-s^2)F(-t^2)\frac{1}{(s+t)^2}dsdt.\nonumber
\eea
Noting that
$$
\frac{1}{(s+t)^2}=-\frac{1}{2}\int_{-\infty}^{\infty}|x|\exp(-ix(s+t))dx.
$$
We have
\bea
S_{1}
&\sim&-\frac{\lambda^2}{16\pi^2}\int_{-\infty}^{\infty}\int_{-\infty}^{\infty}\int_{-\infty}^{\infty}F(-s^2)F(-t^2)|x|\exp(-ix(s+t))dxdsdt\nonumber\\
&=&-\frac{\lambda^2}{16\pi^2}\int_{-\infty}^{\infty}|x|\left[\int_{-\infty}^{\infty}F(-s^2)\exp(-ixs)ds\right]^2dx\nonumber\\
&=&-\frac{\lambda^2}{2\pi^2}\int_{0}^{\infty}x\left[\int_{0}^{\infty}F(-s^2)\cos(xs)ds\right]^2dx.\nonumber
\eea
Now we compute $S_{2}$. Let $y=\sqrt{2N}-x$,
$$
S_{2}=\frac{\lambda}{\pi}\int_{0}^{2\sqrt{2N}}\sqrt{y\left(2\sqrt{2N}-y\right)}F(-y)dy.
$$
As $N\rightarrow\infty$,
\bea
S_{2}&\sim&\frac{\sqrt{2\sqrt{2N}}\lambda}{\pi}\int_{0}^{\infty}\sqrt{y}F(-y)dy\nonumber\\
&=&\frac{2^{\frac{3}{4}}N^{\frac{1}{4}}\lambda}{\pi}\int_{0}^{\infty}\sqrt{y}F(-y)dy.\nonumber
\eea
Theorem \ref{mv} then follows.

\section{Conclusion}
This paper studies the large $N$ behavior of the MGF of the scaled linear statistics in Gaussian ensembles and Laguerre ensembles, from which we obtain the mean and variance of the corresponding linear statistics. We find that there is an equivalence between the mean and variance of suitably scaled linear statistics in Gaussian and Laguerre ensembles. In addition, we use the results of \cite{Basor1999} and \cite{Basor2001} to consider another type of linear statistics in GUE and also obtain the mean and variance of the corresponding linear statistics. For the GSE and GOE, we will deal with the corresponding type of linear statistics in the future.

\section*{Acknowledgments}
Chao Min was supported by the Scientific Research Funds of Huaqiao University under grant number 600005-Z17Y0054.
Yang Chen was supported by the Macau Science and Technology Development Fund under grant numbers FDCT 130/2014/A3, FDCT 023/2017/A1 and by the University of Macau under grant numbers MYRG 2014-00011-FST, MYRG 2014-00004-FST.

\end{document}